\documentclass[11pt,reqno]{amsart}

\usepackage[english]{babel}
\usepackage{subcaption, caption, graphicx} 
\usepackage[utf8]{inputenc}
\usepackage{bbm}
\usepackage{amscd}
\usepackage{amsmath,amssymb}

\usepackage{hyperref}
\usepackage{url}

\numberwithin{equation}{section}
\usepackage[margin=2.9cm]{geometry}

\newcommand{\bP}{\mathbb{P}}

\newcommand{\bE}{\mathbb{E}}
\newcommand{\bR}{\mathbb{R}}

\newcommand{\eps}{\varepsilon}

\newtheorem{theorem}{Theorem}[section]
\newtheorem{lemma}[theorem]{Lemma}
\newtheorem{prop}[theorem]{Proposition}
\newtheorem{assumption}{Assumption}
\newtheorem{corollary}[theorem]{Corollary}

\theoremstyle{definition}
\newtheorem{definition}[theorem]{Definition}

\theoremstyle{remark}
\newtheorem{remark}[theorem]{Remark}

\def\eqref#1{(\ref{#1})}

\begin{document}

\title[Optimal Portfolio in Intraday Electricity Markets]{Optimal Portfolio in Intraday Electricity Markets Modelled by L\'{e}vy-Ornstein-Uhlenbeck Processes}
\thanks{The authors wish to thank Alessandro Balata, Fred Espen Benth, Sebastiano Don, Giorgio Ferrari, Paolo Luzzini, Frank Norbert Proske, Wolfgang J. Runggaldier and all the participants of the XVIII Workshop in Quantitative Finance in Milan, January 25-27, 2017 for useful discussion and suggestions. This work was partly supported by the grant CPDA158845-2015 of the University of Padua ``Multidimensional polynomial processes and applications to new challenges in mathematical finance and in energy markets".}

\author[Marco Piccirilli]{Marco Piccirilli}
\address[Marco Piccirilli]{University of Padova, Department of Mathematics, via Trieste 63, Torre Archimede, I-35121 Padova (Italy)}
\email[Corresponding author]{mpicciri@math.unipd.it}
\author[Tiziano Vargiolu]{Tiziano Vargiolu}
\address[Tiziano Vargiolu]{University of Padova, Department of Mathematics, via Trieste 63, Torre Archimede, I-35121 Padova (Italy)}
\email{vargiolu@math.unipd.it}

\date{\today}


\begin{abstract}
We study an optimal portfolio problem designed for an agent operating in intraday electricity markets. The investor is allowed to trade in a single risky asset modelling the continuously traded power and aims to maximize the expected terminal utility of his wealth. We assume a mean-reverting additive process to drive the power prices.
In the case of logarithmic utility, we reduce the fully non-linear Hamilton-Jacobi-Bellman equation to a linear parabolic integro-differential equation, for which we explicitly exhibit a classical solution in two cases of modelling interest. The optimal strategy is given implicitly as the solution of an integral equation, which is possible to solve numerically as well as to describe analytically. An analysis of two different approximations for the optimal policy is provided. Finally, we perform a numerical test by adapting the parameters of a popular electricity spot price model.
\end{abstract}

\maketitle

\section{Introduction}

After power markets have been deregulated worldwide, this financial sector is experiencing profound structural changes. Several interesting issues have triggered a growing interest in models for energy markets. A rigorous mathematical approach may be useful for practitioners and at the same time stimulate the advance in academic research.
In this work we consider markets structured as the European Power Exchange (EPEX), which regulates electricity spot trading in Central and Western Europe. 
In many exchanges, short-term trade is organized in mainly two markets: day-ahead and intraday. While the day-ahead market trades electricity for each hour (or block of hours) to be delivered the next day and is auction-based, the intraday market, which opens after the day-ahead closure, offers the participants the possibility to trade continuously until short time prior to delivery. The intraday market is especially important for renewable energy producers, who can adjust their day-ahead positions due to changes of weather forestcasts \cite{kieselpara}. 

In this work we study a dynamic portfolio optimization problem designed for intraday electricity trading. These markets play an important role in the equilibrium of power grids, since both electricity producers and consumers are allowed to optimize their positions and reduce the risk of imbalance, which entails fees to be paid to system  operators. Since the increasing penetration of renewable sources, modelling intraday trading has become particularly important, as well as mathematically interesting. Literature related to this problem is very recent: one of the first paper in this direction is \cite{henriot}, where the author studies how a wind power producer may benefit from trading in intraday markets by taking into account the risk of forecast production errors. 
\cite{garnier} study an optimal trade execution problem in order to compensate forecast errors on wind or photovoltaic power production. 
\cite{MR3439396} consider a producer who aims to minimimize the imbalance cost of his residual demand (which is stochastic) by controlling his flexible power production (thermal plants) and his position on the intraday market. Another recent paper \cite{farinelli} studies a stochastic multiperiod portfolio optimization problem in discrete time for hydroassets management and derives, in particular, an optimal intraday trading strategy. In \cite{kieselpara} the authors investigate the impact of intraday updated forecasts of wind and photovoltaic on the bidding behavior of market participants, while \cite{wolff} compare the price drivers in both the EPEX day-ahead and intraday electricity markets. Also, \cite{cartea_flora} study cross-border effects in intraday prices between interconnected locations and \cite{wind} consider a wind energy producer who trades in forward, spot, intraday and adjustment markets and derive optimal trading policies taking into account that his forecast production is imperfect. Our study arises as natural generalization of \cite{gallana} and takes the perspective of a small agent interested in exploiting the stylized features of intraday prices in order to maximize his expected terminal gains.  

We propose a stochastic model for the continuously traded electricity price and formulate an expected profit maximation problem in the language of stochastic control. The price is modeled by additive non-Gaussian Ornstein-Uhlenbeck (OU) processes. Power spot prices are usually described by either geometric \cite{carteaF,gemanR} or additive mean-reverting processes \cite{MR2323278,BPV,koeke,latini,Lucia2002}. In the context of intraday markets, our model choice generalizes both \cite{gallana}, where the authors model the price by Gaussian OU processes, and \cite{kieselpara}, where the intraday price is an AR(1) process (with regime-switching that we do not consider), which is the discrete time version of the OU process. 
This class is highly flexible and capable to reproduce both the mean-reverting and the spiky behavior of observed time series. Since we do not perform a logarithmic transformation of the price, as is usually done, negative prices can be reproduced by our model: this is consistent with what has been recently observed since the introduction of renewable energy sources in the power mix (see \cite{negative_prices}).

Following the pioneering work of Merton \cite{merton,MR0456373}, optimal portfolio management has become one of the most popular problems in mathematical finance and has been addressed in several different frameworks. For a general treatment we refer the reader to \cite{MR2322248}. Our approach is based on the dynamic programming method and the study of the associated Hamilton-Jacobi-Bellman (HJB) integro-differential equation. Some related works include e.g. \cite{MR2153580}, where the price evolves as an exponential Ornstein-Uhlenbeck process, known as the Schwartz model, which is ubiquitous in commodity prices modelling. 
In \cite{kallsen2000}  the authors model the risky assets with exponential Lévy processes and \cite{pasin} generalize to exponential additive processes. We are inspired by \cite{MR1922696,MR1732400} when introducing a transformation for solving the HJB equation. Also, \cite{MR1967775} study the same optimization problem for the Barndorff-Nielsen-Shephard model \cite{MR1841412}, where the volatility is a superposition of non-Gaussian OU processes driven by subordinators. 

The problem of finding the value function is not straightforward from the formulation of the HJB equation. We study the case of logarithmic utility in order to disentangle the terms depending on both the wealth process and the strategy, from the ones depending on time and price. This simplification was also observed by \cite{aase} for a certain class of jump-diffusion processes. A theoretical study of optimal portfolios in the case of logarithmic utility is performed in \cite{goll_kallsen} and furtherly generalized in \cite{goll_kallsen2003}, where the authors apply martingale methods (see \cite{goll_kallsen,goll_kallsen2003} and references therein for details on this approach) to a general semimartingale framework. However, although their analysis provides a characterization of an optimal strategy and its uniqueness,  
the study of its analytical properties is not explicitly addressed. We reduce the fully nonlinear HJB equation to a linear partial integro-differential equation (PIDE) by applying a logarithmic transformation as in \cite{MR1922696}. Even though the optimal strategy is given implicitly as the solution of an integral equation, we are able to show that it is well-defined and satisfies sufficiently regular properties in order to apply the Verification Theorem. 

We prove the existence of a classical solution to the HJB equation in two cases of interest: time-inhomogeneous compound Poisson processes with non-degenerate Brownian component and additive pure-jump processes of (possibly) infinite variation. This is done in the first case relying on a result by \cite{MR1650147}, while in the second case via Feynman-Ka\v{c} representations. In the latter approach, we follow the idea of \cite{MR1967775} and generalize it to time-inhomogeneous processes, which we do not assume  to be of finite variation as in \cite{MR1967775}. In particular, Danskin's theorem \cite{MR0210456} allows us to prove that the forcing term of the HJB equation, which is defined as the composition of non-differentiable functions, is actually differentiable. 

Partial integro-differential equations (PIDE) are in itself of interest and arise across different fields of mathematics. 
In our paper we consider classical solutions, obtained via probabilistic representations, partly as extensions or complementary contributions of various earlier works. A classical reference for this type of problems is \cite{MR673169}, where some existence results are stated under strong regularity assumptions on the coefficients of the equation. In proving the existence of a regular solution for finite Lévy measures with non-degenerate Brownian component, we apply a result of \cite{MR1650147}. Nevertheless, this approach is based on classical smoothness results from PDE theory for linear second-order differential equations (cf. \cite{MR0181836}), which require the finiteness of the jump measure. Consequently, for more general jump-processes we instead follow \cite{MR1967775}, where the Feynman-Ka\v{c} formula yields a candidate, which is proven to be a classical solution of a PIDE very similar to ours. Unfortunately, this approach works only in the first-order case, i.e. with no Brownian component. However, as observed by \cite{MR2211710}, in order to generate realistic price trajectories, it is sufficient to consider financial models which are either finite activity jump processes combined with a diffusion part, or infinite activity pure-jump models, since the latter behave in a ``diffusive'' way when frequent small jumps occur. 

We then study an approximation of the optimal strategy based on the Taylor expansion of the first-order condition, which is a numerical integral equation. In the case of compound Poisson processes, the center of the polynomial is chosen as the mean jump size. We compare it to the classical Merton ratio \cite{merton}, which is shown to correspond to a Taylor expansion around zero. This approximation has been studied also in \cite{ascheberg,MR3176490,imma_reno,pasin} for stochastic volatility price models with jumps. Nevertheless, in their approach the authors start by approximating the HJB equation directly, while we work on the first-order condition. We derive some estimates of the approximation error and perform a numerical test on a power spot price model, specifically the factor model in \cite{Benth20121589}. Our main finding here is that Merton's ratio performs poorly in comparison to our jump-based approximation, suggesting that optimal trading in Lévy-driven markets is not well described by this economically meaningful quantity, at least for the specific case of additive mean-reverting prices. 
 
The paper is structured as follows. In Section 2 we introduce the intraday price dynamics and set our stochastic control problem. In Section 3 we describe the properties of our optimal strategy and study the reduced HJB equation for a logarithmic utility. In particular, two existence results of classical solutions to PIDE are given. We conclude this section by applying the Verification Theorem. The approximation study of the optimal strategy is contained in Section 4, while Section 5 presents an exemplary numerical test on the policy approximations. Appendix A includes auxiliary propositions for the existence of a PIDE solution in the first-order case, while in Appendix B we collect some of the most technical proofs.

\section{The Optimal Portfolio Problem}
We follow the dynamic programming strategy for solving our stochastic optimal control problem (see, for instance, \cite{MR2179357}). The purpose is maximizing the expected utility of our portfolio over a set of trading strategies, that is to study the quantity
\begin{equation}\label{21}
\sup_{\pi}\,\bE[U(X(T))],
\end{equation}
where $U:\bR\to\bR$ is a utility function representing the risk profile of the investor, $X=X^{\pi}$ denotes the portfolio value associated to the strategy $\pi$ and $T$ is the trading closure time.

Let us introduce the stochastic dynamics driving the market. Denote by $L$ a real-valued additive process (for details see e.g. \cite[Section 14.1]{MR2042661}) defined on the complete filtered probability space $(\Omega, \mathcal{F}, (\mathcal{F}_{u})_{u\geq0},\bP)$ by
\begin{equation}\label{L}
dL(u)  = b(u)\,du + \sigma(u)\,dW(u) + \psi(u)\,\int_\bR y\,\overline N(dy,du),\qquad u\in(0,T],
\end{equation}
where $\sigma,b,\psi:[0,T]\to\bR$ are continuously differentiable functions, such that $0\leq\psi_1\leq\psi(u)\leq\psi_2$ for any $u\in[0,T]$ and $\sigma(u)$ and $\psi(u)$ do not vanish at the same time. The process $W$ is a standard Brownian motion and $\overline N(dy,du):=N(dy,du)-\nu(dy)\, du$ is the compensated Poisson random measure associated to a Lévy measure $\nu$, i.e. a Radon measure on $\bR\setminus\{0\}$ such that $\int_\bR (1\wedge y^2)\, \nu(dy)<\infty$. In particular, if $b,\sigma$ and $\psi$ are constant, with $\psi\equiv1$, then $L$ is a Lévy process. In order to deal with processes with finite second moment, we furtherly assume that $\nu$ satisfies the following integrability condition:
\begin{equation}\label{variance}
\int_{|y|\geq1} y^2\, \nu(dy)<\infty.
\end{equation}
Observe that $L$ can be decomposed in a deterministic drift part, a Brownian motion with time-varying volatility and a square integrable pure-jump martingale component. 
We also introduce the following convention. 
\begin{remark}[Jump measure support]\label{12}
It holds that supp $\nu\subset[m,M]$ for $-\infty\leq m\leq M\leq+\infty$. We interpret the case $m=M=0$ formally as the diffusive case, i.e. when $L$ has no jump component. If, for instance, $m=-\infty$, we mean $[m,M]:=(-\infty,M]$. 
\end{remark}
We are in a market with one asset (i.e. the continuously traded intraday price of electricity, expressed in Euros per MWh), whose market value $S=S^{t,s}$ evolves in time according to the stochastic differential equation
\begin{align}\label{5}
    dS(u)  &= -\lambda S(u)\,dt + dL(u), & u\in(t,T],  
\end{align}
given the initial condition $S^{t,s}(t)=s$ for some $t\in[0,T)$ and $s\in\bR$. The constant $\lambda$ is positive and represents the mean-reversion rate of $S$. In particular, for any additive process $L$
there exists a unique (strong) solution $S$ such that in general $\bP(S(u)<0)>0$ for some $u\geq t$, i.e. the price may assume negative values. Nevertheless, if $L$ 
can have only positive jumps and there is no Brownian component, i.e. it is a time-inhomogeneous subordinator (see e.g. \cite{MR2042661}), then a nonnegative initial condition $s\geq0$ naturally implies that $S(u)$ is a.s. nonnegative for each $u\geq t$. Therefore, it is possible to consider additive processes taking both positive and negative values, according to one's modelling preferences. 
The unique solution of \eqref{5}, with starting condition $S^{t,s}(t)=s$, can be explicitly written as
\begin{equation}\label{9}
S^{t,s}(u)=s e^{-\lambda (u-t)}+\int_t^u e^{-\lambda (u-v)}\, dL(v).
\end{equation}
Since $L$ has finite second moment by \eqref{variance}, $S$ has finite second moment as well. 

If $\pi(u)$ represents the amount of shares of the stock (i.e. the amount of energy in MWh), owned by the agent at time $u$, the associated self-financing portfolio dynamics $X=X^{t,s,x;\pi}$ is described by
\begin{align}\label{1}
    dX(u)  &= \pi(u) \,dS(u),& u\in(t,T],  \\ 
    X(t) &= x,
\end{align}
where $S=S^{t,s}$ and $x>0$. We also assume that $\pi(T)=0$, so that the agent liquidates the position at the terminal time $T$; in other words, we are in a pure trading context (see \cite[Remark 3.1]{gallana}). We then define the set of admissible trading strategies and the value function.
\begin{definition}[Admissible controls]\label{11}  
We call $\mathcal{A}([t,T])$ the set of admissible controls, which are defined as real-valued predictable processes $\pi$ on $[t,T]$ (in the sense of \cite[Definition 3.3]{MR637061}) such that the following conditions hold:
\begin{enumerate}
\item Equations \eqref{5} and \eqref{1} admit a unique strong solution $(S,X)=(S^{t,s,x},$ $X^{t,s,x;\pi})$ for each initial condition $S(t)=s,X(t)=x$, with $t\in[0,T)$ and $(s,x)\in\bR\times\bR^+$.
\item The associated wealth process is positive, i.e. $X^{t,s,x;\pi}(u)>0$, $\bP$-a.s. for each $u\in(t,T]$ and the final net position is zero: $\pi(T)=0$.
\end{enumerate}
\end{definition}
\begin{definition}[Value function]
If $(S^{t,s,x},X^{t,s,x;\pi})$ denotes the controlled Markov process starting from $(s,x)$ at time $t$ and evolving as in \eqref{5} and \eqref{1}, we define the value function by
$$
V(t,s,x)=\sup_{\pi\in\mathcal{A}([t,T])} J(t,s,x;\pi),
$$ 
where $J$ is the objective function:
$$
J(t,s,x;\pi)=\bE[U(X^{t,s,x;\pi}(T))].
$$
The function $U:\bR^+\to\bR$ represents the investor's utility and is concave, increasing,
and bounded from below.
\end{definition}
Following the dynamic programming principle, the Hamilton-Jacobi-Bellman (HJB) equation associated to this optimization problem is 
\begin{align}\label{2}
\frac{\partial}{\partial t}H(t,s,x)+\sup_\pi A^{\pi} H(t,s,x)&= 0, &(t,s,x)\in [0,T)\times\bR\times\bR^+,\\
\label{3}
H(T,s,x)&= U(x), &(s,x)\in\bR\times\bR^+.
\end{align}
According to \eqref{5} and \eqref{1}, the infinitesimal generator $A^{\pi}$ of the controlled process $(S^{t,s,x},$ $X^{t,s,x;\pi})$ acts on a sufficiently regular function $H(t,s,x)$ as follows
\begin{align*}
A^{\pi}H(t,s,x)&=(b(t)-\lambda s) (\pi H_x (t,s,x)+ H_s (t,s,x))\\
&+\frac12 \pi^2\sigma(t)^2 H_{xx} (t,s,x)+\pi \sigma(t)^2 H_{sx} (t,s,x)+\frac12 \sigma(t)^2 H_{ss} (t,s,x)
\\
&+\int_{\bR^2} [H(t,s+y_2,x+\pi \psi(t) y_1)-H(t,s,x)
\\
&-
(\pi H_x (t,s,x)\,\psi(t) y_1+ H_s (t,s,x) \psi(t) y_2)]\, \tilde\nu_t(dy_1 dy_2),
\end{align*}
where $\tilde\nu_t$ is the jump measure associated to the two-dimensional process $(S,X)$. Since this is a singular two-dimensional measure which coincides with the one-dimensional jump measure $\nu(dy)$ on the line $y_1=y_2$, we can rewrite the integral term as
\begin{align*}
A^{\pi}H(t,s,x)
&=(b(t)-\lambda s) (\pi H_x (t,s,x)+ H_s (t,s,x))\\
&+\frac12 \pi^2\sigma(t)^2 H_{xx} (t,s,x)+\pi \sigma(t)^2 H_{sx} (t,s,x)
+\frac12 \sigma(t)^2 H_{ss} (t,s,x)
\\
&+\int_\bR [H(t,s+\psi(t) y,x+\pi \psi(t) y)-H(t,s,x)
-(\pi H_x (t,s,x) + H_s (t,s,x))\, \psi(t) y]\, \nu(dy).
\end{align*}
To link the HJB equation to the control problem, we formulate a Verification Theorem in the version of \cite[Theorem III.8.1]{MR2179357}. The basic tool is the well-known Dynkin formula (see \cite[p.122]{MR2179357}), which here applies to the controlled process $(S,X^\pi)$:
\begin{equation}\label{4}\begin{split}
\bE^{t,s,x}[f(T,S(T),X^\pi(T))]-f(t,s,x)
=\bE^{t,s,x}\left[\int_t^T A^{\pi(u)} f(u,S(u),X^\pi(u))\, du\right],
\end{split}\end{equation}
where $\bE^{t,s,x}$ denotes the conditional expectation given $S(t)=s,$ $X^\pi(t)=x$ and $f:[0,T]\times \bR\times\bR^+\to \bR$ is any function for which the expression makes sense.

\begin{theorem}[Verification Theorem]\label{verification}
Define the set 
$$
\mathcal{D}=\{f\in C^{1,2}([0,T)\times\bR\times\bR^+) \mbox{ so that } \eqref{4} \mbox{ holds for each }\pi\in \mathcal{A}([t,T])\}.
$$
Let $H\in\mathcal{D}$ be a classical solution of \eqref{2} which respects the terminal condition \eqref{3}. Then it holds, for each $(t,s,x)\in[0,T]\times\bR\times\bR^+$,
\begin{enumerate}
\item $H(t,s,x)\geq J(t,s,x;\pi)$ for each admissible control $\pi\in\mathcal{A}[t,T]$;
\item if there exists an admissible control $\pi^*\in\mathcal{A}[t,T]$ such that
$$
\pi^*(u)\in\arg\max_{\pi} A^{\pi}H(u,S(u),X^\pi(u)) \qquad \bP\mbox{-a.s.}\,\mbox{ for }\, u\in[t,T]
$$
then $H(t,s,x)= J(t,s,x;\pi^*)=V(t,s,x)$, i.e. $\pi^*$ is an optimal strategy.
\end{enumerate}  
\end{theorem}
\begin{proof}
The proof is classical and follows directly from the Dynkin formula in \eqref{4}.
\end{proof}

\section{Optimal Control and Value Function for a Logarithmic Utility}

In this section we solve the optimization problem in the case of a logarithmic utility, i.e. when the utility function in \eqref{21} is $U(x)=\log(x)$. Specifically, we find an explicit solution for the HJB equation by means of a logarithmic transform. First, we reduce the fully nonlinear HJB equation to a linear parabolic integro-differential equation for which, under certain assumptions, the existence of a regular solution can be proven. By applying the Verification Theorem of the previous section, we prove it to be equal to the value function of the original maximization problem (Theorem \ref{36}). We also state the existence and give a representation of an optimal strategy, which is shown to solve an integral equation. 
\subsection{Optimal strategy}
By the properties of logarithmic utility, it holds that, if an optimal strategy $\pi^*$ exists, it takes the form 
\begin{equation}\label{fact}
\pi^*(u)=\bar\pi^*(u) X(u-),
\end{equation}
where $\bar\pi^*$ is a predictable process that can be implicitly defined in terms of the semimartingale characteristics of $S^{t,s}$, i.e. on the local behavior of the price process (see \cite[Theorem 3.1]{goll_kallsen}). 
This implies that we may explicitly characterize the strategies for which the wealth process $X$ is positive. In fact, for an admissible strategy $\pi(u):=\bar\pi(u) X(u-)$ we can rewrite \eqref{1} as
\begin{equation}\label{22}
dX(u)  = X(u-)\,\bar\pi(u) \,dS(u).
\end{equation}
This allows us to have for general $\bar\pi$ an explicit formula for $X$, since it takes the form of a stochastic exponential (cf. \cite[Section 8.4]{MR2042661}). By It\^o's formula, 
$$
X(u)=x\cdot e^{\bar\pi(u) S(u)-\frac{1}2 \int_t^u \sigma^2(v)\bar\pi^2(v) \,dv} \prod_{t<v\leq u} (1+\bar\pi(v)\Delta S(v)) e^{-\Delta S(v)}, \quad \bP\mbox{-a.s.}
$$
\begin{prop}[Positivity of the portfolio value]
If $\pi(u)=\bar\pi(u) X(u-)$, it holds $X^{t,s,x;\pi}(u)>0$, $\bP$-a.s., $\forall u\in[t,T]$, if and only if 
$$
\bar\pi(u)\psi(u) y>-1\quad \bP\mbox{-a.s.}, \ \nu\mbox{-a.e.}\ y\in\bR, \ \mbox{for all}\ u\in[t,T].
$$ 
\end{prop}
\begin{proof}
From \eqref{1}, if the jump measure at time $t$ of $S$, regarded as an additive process, is denoted by $\nu_t^S$, then it holds that supp $\nu_t^S=$ supp $\nu$. Then, $\{\bar\pi(u)\Delta S(u)>-1$ $\bP$-a.s., $\forall u\leq T$\} if and only if \{$\bar\pi(u)\Delta L(u)>-1$ $\bP$-a.s., $\forall u\leq T$\}, which is equivalent to \{$\bar\pi(u)\psi(u) y>-1$ $\bP$-a.s., $\nu$-a.e. $y\in\bR$, $\forall u\leq T$\}.
\end{proof}
Therefore, the portfolio is positive for each strategy of the form $\pi(u)=\bar\pi(u) X(u-)$ such that $\bar\pi$ takes values in a suitably chosen set. 
This sums up in the following characterization of admissible controls. 
\begin{definition}\label{18bis}
Let $\Pi=\Pi_{\nu,\psi}$ be a compact set such that 
$$
\Pi_{\nu,\psi}\subset\widehat\Pi_{\nu,\psi}:=\{\bar\pi\in\bR \mbox{ s.t. }\bar\pi \psi y>-1 \mbox{ for each } y\in[m,M] \mbox{ and } \psi\in[\psi_1,\psi_2] \}. 
$$
A predictable process $\bar\pi:[t,T]\to\Pi$ is called normalized admissible strategy if there exists an admissible strategy $\pi\in\mathcal{A}([t,T])$ such that
$$
\pi(u)=\bar\pi(u) X(u-)
$$ 
for all $u\in[t,T]$, $\bP$-a.s.
\end{definition}
\begin{remark}\label{18}
According to the support of the measure $\nu$, the set $\widehat\Pi:=\widehat\Pi_{\nu,\psi}$ consists of 
\begin{description}
\item[case A] $\widehat\Pi=(-\frac1{M\psi_2},-\frac1{m\psi_2})$ if $m<0$ and $M>0$ (both positive and negative jumps),
\item[case B]$\widehat\Pi=(-\frac1{M\psi_2},+\infty)$ if $0\leq m\leq M$ and $M\not=0$ (only positive jumps),
\item[case C] $\widehat\Pi=(-\infty,-\frac1{m\psi_2})$ if $m\leq M\leq 0$ and $m\not=0$ (only negative jumps),
\item[case D] $\widehat\Pi=\bR$ if $m=M=0$ (no jumps),
\end{description}
by consistently interpreting where necessary: for instance if $M=+\infty$, $(-\frac1{M\psi_2},+\infty):=[0,+\infty)$. Observe that in all cases we have $0\in\widehat\Pi$. If $m=-\infty$ and $M=\infty$, then $\widehat\Pi=\{0\}$ which makes the problem trivial. Therefore, in order to get rid of this situation, we may assume from now on that at least one between $m$ and $M$ is finite.
\end{remark}

\begin{remark}\label{23}
The set $\widehat\Pi=\widehat\Pi_{\nu,\psi}$ is defined according to the jump features of the process $L$  (cf. the analogous notion of \emph{neutral constraints} in \cite[Section 2]{goll_kallsen2003}). On the other hand, we have a certain freedom in the definition of $\Pi$, as we only require that it is a compact subset of $\widehat\Pi$. Intuitively, we are restricting the range of possible trading strategies so that the instantaneous portfolio value can not jump to (or below) zero for any admissible (normalized) position $\bar\pi$.
\end{remark}

In order to find a solution to the HJB equation, we make the following ansatz:
$$
H(t,s,x)=U(x\,e^{g(t,s)})=\log(x)+g(t,s).
$$ 
This transform, which has been introduced in \cite{gallana} for the specific case of Gaussian processes, is analogous to the one employed in \cite{MR1922696}, with the main difference due to the arithmetic nature of our spot price dynamics.
We start from the static maximization problem, namely the maximization of the generalized Hamiltonian over all possible values of the strategies $\pi$. 
As usual in this approach (see the discussion in \cite{MR2179357}) a candidate optimal policy $\pi\in\mathcal{A}[t,T]$ can be found by computing $\pi^*(t,s,x)=\arg\max_{\pi} A^{\pi}H(t,s,x)$ and defining $\pi^*(t):=\pi^*(t,S(t-),X(t-))$. It is common to refer to the deterministic function $\pi^*(t,s,x)$ as the optimal Markov control policy. Since we are in the case of logarithmic utility (cf. \eqref{fact}), we can write $\pi^*(t,s,x)=\bar\pi^*(t,s)\cdot x$. Simple computations yield
\begin{align}\label{27}
A^{\pi}H(t,s,x)&=(b(t)-\lambda s) \left(\frac\pi x+ g_s(t,s)\right)-\frac12\,\sigma(t)^2\frac{\pi^2}{x^2}+\frac12 \sigma(t)^2 g_{ss} (t,s)
\\ \nonumber
&+\int_{\bR} \Bigl[\log(x+\pi \psi(t) y)+g(t,s+y)-\log(x)-g(t,s)
-\left(\frac{\pi \psi(t)}x+ g_s(t,s)\right) y\Bigr]\nu(dy).
\end{align}
Neglecting the terms which do not depend on $\pi$, we have
\begin{align*}
\arg\max_\pi A^{\pi}H(t,s,x)&=\arg\max_\pi \ (b(t)-\lambda s) \frac\pi x-\frac12\,\sigma(t)^2\frac{\pi^2}{x^2}+\int_\bR \left[\log\left(1+\frac{\pi \psi(t)}x y\right)-\frac{\pi \psi(t)}x y\right]\, \nu(dy)\\
&=x\cdot\arg\max_{\bar\pi\in\Pi} f(\bar\pi;t,s),
\end{align*}
where the function $f:\Pi\times[0,T]\times\bR\to\bR$ is defined as
\begin{equation}\label{25}
f(\bar\pi;t,s):=(b(t)-\lambda s) \bar\pi-\frac12\,\sigma(t)^2 \bar\pi^2+\int_\bR \left[\log\left(1+\bar\pi \psi(t) y\right)-\bar\pi \psi(t) y\right]\, \nu(dy).
\end{equation}
The expression to be maximized with respect to the variable $\bar\pi$ reads as the sum of three terms: a linear term, a strictly concave function and the integral of a strictly concave function. Therefore we are maximizing an overall strictly concave function on a compact set $\Pi$. This ensures the existence of a unique maximizer 
\begin{equation}\label{24}
\bar\pi^*=\bar\pi^*(t,s):=\arg\max_{\bar\pi\in\Pi} f(\bar\pi;t,s).
\end{equation}
\begin{remark}
By adopting this notation we are revealing in advance that $\pi^*$ corresponds to an optimal strategy, but we have not given a proof yet. The optimality of this candidate will be derived in Theorem \ref{36} by applying the Verification Theorem of the previous section.
\end{remark}
Recalling that $\pi^*(t,s,x)=\bar\pi^*(t,s)\cdot x$,
we can write the HJB equation in reduced form 
$$
\frac{\partial}{\partial t} H(t,s,x)+A^{\pi^*(t,s,x)} H(t,s,x)= 0,
$$ 
that is, consistently with our guess $H(t,s,x)=\log(x)+g(t,s)$,
\begin{align*}
g_t(t&,s)+(b(t)-\lambda s) \left(\bar\pi^*(t,s)+ g_s(t,s)\right)-\frac12\,\sigma(t)^2 \bar\pi^*(t,s)^2
+\frac12 \sigma(t)^2 g_{ss} (t,s)
\\
+&\int_\bR \Bigl[\log\left(1+\bar\pi^*(t,s) \psi(t) y\right)-\bar\pi^*(t,s) \psi(t) y+g(t,s+ \psi(t)y)-g(t,s)-g_s(t,s)  \psi(t) y\Bigr]\, \nu(dy)=0,
\end{align*}
After the terms with $g$ are collected, the equation reads
\begin{align}\nonumber
g_t(t,s)&+(b(t)-\lambda s) g_s(t,s)+\frac12 \sigma(t)^2 g_{ss} (t,s)\\ \label{97}
&+\int_\bR \left[g(t,s+\psi(t) y)-g(t,s)-g_s(t,s) \psi(t) y\right]\, \nu(dy)=-f^*(t,s),
\end{align}
with terminal condition $g(T,s)=0$, 
where we define $f^*:[0,T]\times\bR\to\bR$ by
\begin{equation}\label{98}
f^*(t,s):=f(\bar\pi^*(t,s);t,s).
\end{equation}

If we interpret \eqref{97} as an equation in the only unknown $g$, it takes the form of a linear parabolic partial integro-differential equation (PIDE). 
The analysis of such an equation is typically a delicate task
and, to the best of our knowledge, there are not many existence results for \emph{regular} solutions in literature for this class of problems (see \cite{MR673169,MR2211710,MR2418232,MR1650147} and references therein). Under certain assumptions, we are able to prove the existence and a probabilistic representation formula: we will do this in Propositions \ref{z} and \ref{41}. 
It is crucial noticing that in the logarithmic case we can solve the HJB equation directly by disentangling the problem of finding $\bar\pi^*(t,s)$ and the function $g(t,s)$. This has been verified by the authors not to be the case for a general CARA or CRRA utility, which makes the issue of solving the HJB equation more difficult as well as interesting (see also \cite{aase}). Nevertheless, an approximation of the HJB equation has been proposed in an analogous stochastic framework for CRRA utility in \cite{ascheberg,imma_reno}. 

In order to solve the PIDE, we first have to study the properties of the strategy defined implicitly in \eqref{24}. A straightforward application of the dominated convergence theorem and the finiteness of the second moment of $L$ assure that $f(\cdot;t,s)$ is differentiable for any $t\in[0,T]$ and $s\in\bR$. Therefore, if the maximizer $\bar\pi^*=\bar\pi^*(t,s)$ is a internal point, it is the unique solution of the first order condition
\begin{equation}\label{7}
f'(\bar\pi^*;t,s)=b(t)-\lambda s -\sigma(t)^2\bar\pi^*-\int_\bR 
\frac{\bar\pi^* \psi(t)^2 y^2}{1+\bar\pi^* \psi(t) y}\, \nu(dy)=0.
\end{equation}
We remark that this is the explicit deterministic counterpart of the third condition appearing in \cite[Theorem 3.1]{goll_kallsen}.

In the two upcoming propositions, we sum up the properties of the candidate (normalized) optimal policy and of the function $f^*$ appearing in the HJB equation.
\begin{prop}\label{19}
Assume that $\Pi$ is a compact interval containing $0$. 
The static optimization problem in \eqref{25} and \eqref{24} admits a unique maximizer $\bar\pi^*:[0,T]\times\bR\to\Pi$ with the following properties:
\begin{enumerate}
\item For each $t\in[0,T]$, it holds $\bar\pi^*\left(t,\frac{b(t)}\lambda\right)=0$.
\item The map $\bar\pi^*:[0,T]\times\bR\to\Pi$ is continuous and then, in particular, measurable and bounded.
\item For each $t\in[0,T]$, there exists an open interval $\Sigma(t)$ such that the restrictions $\bar\pi^*(t,\cdot)|_{\Sigma(t)}$ are strictly decreasing and smooth, where
\begin{description}
\item[case A] $\Sigma(t)=(s_1(t),s_2(t))$,
\item[case B] $\Sigma(t)=(-\infty,s_2(t))$,
\item[case C] $\Sigma(t)=(s_1(t),+\infty)$,
\item[case D] $\Sigma(t)\equiv\bR$.
\end{description}
Moreover, for any $t\in[0,T]$ the derivatives of $\bar\pi^*(t,\cdot)|_{\Sigma(t)}$ can be extended to $\overline\Sigma(t)$. 
\item For each $t\in[0,T]$, the map $\bar\pi^*(t,\cdot):\bR\to\Pi$ is decreasing on the whole real line and, in particular,
\begin{description}
\item[case A] there exist $s_1,s_2$ such that, for any $t\in[0,T]$, we have $-\infty< s_1\leq 
\frac{b(t)}\lambda \leq s_2<\infty$ and
$$
\bar\pi^*(t,s)\equiv\begin{cases} \max\Pi, & \mbox{if }s\leq s_1, \\ 
\min\Pi, & \mbox{if }s\geq s_2.\end{cases}
$$
\item[case B] there exists $s_2$ such that, for any $t\in[0,T]$, we have $\frac{b(t)}\lambda\leq s_2<\infty$ and $\bar\pi^*(t,s)\equiv \min \Pi$ for $s\geq s_2$  .

\item[case C] there exists $s_1$ such that, for any $t\in[0,T]$, we have $-\infty< s_1\leq \frac{b(t)}\lambda$ and $\bar\pi^*(t,s)\equiv \max\Pi$ for $ s\leq s_1$.

\item[case D] we can write down the maximizer explicitly as 
$$
\bar\pi^*(t,s)=\frac{b(t)-\lambda s}{\sigma(t)^2}, 
$$
for each $t$ and $s$ such that the above quantity is well-defined and belongs to $\Pi$.
\end{description}
\item In particular, for all $t\in[0,T]$ the maps $\bar\pi^*(t,\cdot)$ are Lipschitz continuous uniformly in $t\in[0,T]$ (i.e. with Lipschitz constant $L$ independent of $t$).
\end{enumerate}
\end{prop}
\begin{proof}
See Appendix B.
\end{proof}
\begin{remark}
In the notation of Proposition \ref{19}, for {\bf case A,B,C}, we can write more explicitly
$$
s_1(t)=\lim_{\bar\pi\to\bar\pi_{2}} s^*(t,\bar\pi)=\frac1\lambda \left(b(t)-\sigma(t)^2\bar\pi_{2}-\int_\bR 
\frac{\bar\pi_{2} \psi(t)^2 y^2}{1+\bar\pi_{2} \psi(t) y}\, \nu(dy)\right),
$$ 
$$
s_2(t)=\lim_{\bar\pi\to\bar\pi_{1}} s^*(t,\bar\pi)=\frac1\lambda \left(b(t)-\sigma(t)^2\bar\pi_{1}-\int_\bR 
\frac{\bar\pi_{1} \psi(t)^2 y^2}{1+\bar\pi_{1} \psi(t) y}\, \nu(dy)\right),
$$
and
$$
s_1=\min_{t\in[0,T]} s_1(t),\qquad s_2=\max_{t\in[0,T]} s_2(t).
$$
\end{remark}
\begin{remark}
In order to interpret the results of Proposition \ref{19}, let us suppose for example to be in {\bf case A}. Recall that in this case we can have both upward and downward jumps in prices (see Remark \ref{18}) and that the normalized position $\bar\pi(t,s(t))$ can take also negative values, so that short-selling is allowed. At each time $t$, a trader who executes optimally takes a net zero position if price $s(t)$ reaches the (time-dependent) ``equilibrium'' level $b(t)/\lambda$. Furtherly, he goes long if price goes above this level and, accordingly, he goes short when the price is below. The trading allocation increases (with sign) as price decreases and vice versa. Also, $s_1$ (resp. $s_2$) consists of a lower (resp. upper) price threshold at which the trader, independently of the time instant, takes the longest (resp. shortest) position possible according to the trading constraints prescribed in $\Pi$.
\end{remark}
\begin{prop}\label{91}
The function $f^*:[0,T]\times\bR\to\bR$ in \eqref{98} is continuously differentiable and Lipschitz continuous, with partial derivatives 
\begin{align*}
\frac{\partial}{\partial t} f^*(t,s)&= b'(t) \bar\pi^*(t,s)-\sigma(t)\sigma'(t) {\bar\pi^*(t,s)}^2-\psi(t)\psi'(t) {\bar\pi^*(t,s)}^2 
\,\int_\bR \frac{y^2}{1+\bar\pi^*(t,s) \psi(t) y}\, \nu(dy),\\
\frac{\partial}{\partial s} f^*(t,s)&=  -\lambda \bar\pi^*(t,s).
\end{align*}
Furthermore, it grows as a linear function of $s$ uniformly in $t$, i.e. 
$$
|f^*(t,s)|\leq C (1+ |s|),
$$
being $C$ dependent only on $\lambda,T,\|b\|_\infty,\|\sigma\|_\infty,\|\psi\|_\infty$.
\end{prop}
\begin{proof}
Recall that by definition
$$
f(\bar\pi;t,s) = (b(t)-\lambda s) \bar\pi-\frac12\,\sigma(t)^2 \bar\pi^2+\int_\bR \left[\log\left(1+\bar\pi \psi(t) y\right)-\bar\pi \psi(t) y\right]\, \nu(dy),
$$
which is a continuously differentiable function in the variable $(t,s)$ for any $\bar\pi\in\Pi$, since $b,\sigma$ and $\psi$ are continuously differentiable.
Then, by Danskin's theorem \cite[Theorem 1]{MR0210456},
$$
f^*(t,s)=\max_{\bar\pi\in\Pi} f(\bar\pi;t,s)
$$
is differentiable with partial derivatives
\begin{align*}
\frac{\partial}{\partial t} f^*(t,s)&= \frac{\partial}{\partial t} f(\bar\pi;t,s)|_{\bar\pi=\bar\pi^*(t,s)},\\
\frac{\partial}{\partial s} f^*(t,s)&= \frac{\partial}{\partial s} f(\bar\pi;t,s)|_{\bar\pi=\bar\pi^*(t,s)}.
\end{align*}
Since they are bounded continuous functions, it follows that $f^*\in C^{1}([0,T]\times\bR)$ and Lipschitz continuous. The linear bound is direct consequence of the definition of $f$ and the boundedness of $\bar\pi^*(t,s)$.
\end{proof}
\subsection{Probabilistic representation and existence of regular solutions}
After studying the regularity properties of the forcing term of the reduced HJB equation \eqref{97}, we move on to the problem of existence of solutions. 
First of all, we clarify the natural notion of classical solution for such a class of integro-differential equations. Tracing through \cite[Section 17.4]{MR3443368}, we say that a function $g:[0,T]\times\bR\to\bR$ belongs to the set $C_{\nu,\psi}^{1,2}=C_{\nu,\psi}^{1,2}([0,T)\times\bR)$, if it is once continuously differentiable in its first argument and twice continuously differentiable in its second and, furtherly, the following integrability condition holds true: for every $t\in[0,T)$ and $s\in\bR$,
\begin{equation}
\int_\bR \left|g(t,s+\psi(t)y)-g(t,s)-g_s(t,s) \psi(t) y\right|\, \nu(dy)<\infty.
\end{equation}
Then, a classical solution of the HJB equation is a function $g:[0,T]\times\bR\to\bR$ belonging to $C_{\nu,\psi}^{1,2}([0,T)\times\bR)$ and satisfying the integro-differential equation \eqref{97}.

We now present three results. Firstly, we recall a version of the Feynman-Ka\v{c} theorem, which gives the probabilistic representation of regular solutions. Then, we state two existence results for classical solutions: the first is valid for additive processes without diffusion part, while the second works for compound Poisson processes and uniformly non-degenerate Brownian component.

\begin{theorem}[Feynman-Ka\v{c} formula]
Assume that $g$ is a $C_{\nu,\psi}^{1,2}([0,T)\times\bR)\cap C([0,T]\times\bR)$ solution of \eqref{97}, satisfying the growth condition:
$$
\max_{t\in[0,T]} |g(t,s)|\leq K\,(1+ s^2), \qquad \mbox{for } s\in\bR. 
$$
If, moreover, there exists $\eps>0$ such that 
$$
\int_{|y|\geq1} |y|^{2+\eps}\,\nu(dy)<\infty, 
$$
then we can represent $g$ in the following Feynman-Ka\v{c} type form 
\begin{equation}\label{37}
g(t,s)=\bE\left[\int_t^T f^*(u,S^{t,s}(u)) \,du\right].
\end{equation}
\end{theorem}
\begin{proof}
The proof is classical: see \cite[Theorem 17.4.10]{MR3443368}.
\end{proof}

In the upcoming proposition we prove the existence of a classical solution to \eqref{97} in the case that there is no Brownian component. 
\begin{assumption}\label{pure-jump}
The diffusion component in \eqref{L} is identically zero, i.e. $\sigma\equiv0$.
\end{assumption}
We follow the idea of \cite{MR1967775}, where a guess is constructed via the Feynman-Ka\v{c} formula. Let us remark that we generalize the result in \cite{MR1967775} by proving the existence of a classical solution for time-inhomogeneus Lévy processes and possibly infinite variation square integrable Lévy measure.
More in detail, we prove that 
\begin{equation}
G(t,s):=\bE\left[\int_t^T f^*(u,S^{t,s}(u)) \,du\right]
\end{equation}
is a well-defined regular function and solves the PIDE in the classical formulation. We need some preliminary propositions, which are collected in Appendix A. 
\begin{prop}[Pure-jump case]\label{z}
Under Assumption \ref{pure-jump}, the function $G(t,s)$ is continuously differentiable in $t$ for all $s\in\bR$ and solves the following partial integro-differential equation:
\begin{align*}
G_t(t,s)+(b(t)-\lambda s) G_s(t,s)+\int_\bR \left[G(t,s+\psi(t)y)-G(t,s)-G_s(t,s) \psi(t) y\right]\, \nu(dy)=-f^*(t,s),
\end{align*}
with terminal condition $G(T,s)=0$. In particular, $G\in C_{\nu,\psi}^{1,1}([0,T)\times\bR)\cap C([0,T]\times\bR)$. Furthermore, for all $t\in[0,T)$ and $s\in\bR$ the following integrability condition holds:
\begin{align*}
&\bE\left[\int_{t}^T \int_\bR \left[G(u,S^{t,s}(u-)+\psi(u) y)-G(u,S^{t,s}(u-))\right]^2\,\nu(dy) \, du\right]<\infty.
\end{align*}
\end{prop}
\begin{proof}
Fix $t\in[0,T),h>0$ and apply It\^o's Lemma to $G(t+h,S(\cdot))$ from $t$ to $t+h$. Then, we have
\begin{align*}
G(t+h,S(t+h))=G(t+h,S(t))+\int_t^{t+h} (b(u)-\lambda S(u))\partial_s G(t+h,S(u)) \,du\\
+\int_t^{t+h} \int_\bR [G(t+h,S(u)+\psi(u) y)-G(t+h,S(u))-\partial_s G(t+h,S(u))\psi(u) y] \nu(dy)du \\
+\int_t^{t+h} \int_\bR [G(t+h,S(u-)+\psi(u) y)-G(t+h,S(u-))] \overline N(dy,du). 
\end{align*}
Now, divide by $h$ and take expectation $\bE^{t,s}$. Fubini's theorem gives that
\begin{align}
\frac1h &\left(\bE^{t,s}[G(t+h,S(t+h))]-G(t+h,s)\right)\\ \nonumber
&=\frac 1h\int_t^{t+h} \bE^{t,s}[(b(u)-\lambda S(u))\partial_s G(t+h,S(u))] \,du \\ \nonumber
&+\frac1h \int_t^{t+h} \int_\bR \bE^{t,s}[G(t+h,S(u)+\psi(u) y)-G(t+h,S(u))
-\partial_s G(t+h,S(u))\psi(u) y] \,\nu(dy)du \\ \nonumber
&+\frac 1h \bE^{t,s}\left[\int_t^{t+h} \int_\bR \bigl(G(t+h,S(u-)+\psi(u) y)-G(t+h,S(u-))\bigr) \overline N(dy,du)\right]. 
\end{align}
By the mean value theorem, since the map $u\mapsto\bE^{t,s}[(b(u)-\lambda S(u))\partial_s G(t+h,S(u))]$ is continuous, we have for a $u_h\in[t,t+h]$ that
\begin{align}
\frac 1h\Bigl(\int_t^{t+h} \bE^{t,s}[(b(u)-\lambda S(u))\partial_s G(t+h,S(u))] \,du\Bigr)=\bE^{t,s}[(b(u_h)-\lambda S(u_h))\partial_s G(t+h,S(u_h))],
\end{align}
which converges to $(b(t)-\lambda s) G_s(t,s)$ as $h$ approaches $0$. 
Analogously, for the second term it holds that
\begin{align*}
\frac1h \int_t^{t+h} \int_\bR \bE^{t,s}[G(t+h,S(u)+\psi(u) y)-G(t+h,S(u))-\partial_s G(t+h,S(u))\psi(u) y] \nu(dy) du \\
= \int_\bR \bE^{t,s}[G(t+h,S(u_h)+\psi(u_h) y)-G(t+h,S(u_h))-\partial_s G(t+h,S(u_h))\psi(u_h) y] \nu(dy),
\end{align*}
for a $u_h\in[t,t+h]$. Since all the maps in the expectation are continuous, as $h$ tends to zero (cf. Lemma \ref{d}), the last term converges to 
$$
\int_\bR \left[G(t,s+\psi(t)y)-G(t,s)-G_s(t,s) \psi(t) y\right]\, \nu(dy).
$$
Moreover, 
$$
\bE^{t,s}\left[\int_t^{t+h} \int_\bR \bigl(G(t+h,S(u-)+\psi(u) y)-G(t+h,S(u-))\bigr) \overline N(dy,du)\right]=0.
$$
Finally, the left-hand side can be written as
\begin{align*}
\frac1h \left(\bE^{t,s}[G(t+h,S(t+h))]-G(t,s)\right)+\frac1h \left(G(t,s)-G(t+h,s)\right).
\end{align*}
By the Markov property and the tower rule,
\begin{align*}
\bE^{t,s}[G(t+h,S(t+h))]&=\bE\left[\bE\left[\int_{t+h}^T f^*(u,S^{t+h,S^{t,s}(t+h)}(u)) \,du\right]\right]\\
&=\bE\left[\bE\left[\left.\int_{t+h}^T f^*(u,S^{t,s}(u)) \,du\right|\mathcal{F}_{t+h}\right]\right]\\
&=\bE\left[\int_{t+h}^T f^*(u,S^{t,s}(u)) \,du \right].
\end{align*}
Therefore,
\begin{align*}
\frac1h& \left(\bE^{t,s}[G(t+h,S(t+h))]-G(t,s)\right)\\
&=\frac1h \left(\bE\left[\int_{t+h}^T f^*(u,S^{t,s}(u)) \,du \right]-\bE\left[\int_{t}^T f^*(u,S^{t,s}(u)) \,du \right]\right)\\
&=-\frac1h \bE\left[\int_t^{t+h} f^*(u,S^{t,s}(u)) \,du \right],
\end{align*}
which converges to $-f^*(t,s)$ as $h$ goes to zero. Then, we have found that the limit of $-\frac1h \left(G(t+h,s)-G(t,s)\right)$ exists and is equal to
\begin{align*}
(b(t)-\lambda s) G_s(t,s)+\int_\bR \left[G(t,s+\psi(t)y)-G(t,s)-G_s(t,s) \psi(t) y\right]\, \nu(dy)+f^*(t,s),
\end{align*}
so that $G_t(t,s)$ exists and it is continuous, being the right-hand term continuous.  Also, we get from this expression that $G$ solves the integro-differential equation of the statement. 

For the last point, as in Lemma \ref{d}, it is sufficient observe that
\begin{align*}
\bE\left[\bigl(G(u,S^{t,s}(u)+\psi(u) y)-G(u,S^{t,s}(u))\bigr)^2\right] \leq \sup_{z\in\bR} G_s(u,z)^2 \psi(u)^2 y^2\leq C e^{2\lambda u} y^2,
\end{align*}
since $G$ is Lipschitz continuous in $z$ uniformly in $u$.
\end{proof}

In the last proposition, a result by \cite{MR1650147} is applied to prove existence and uniqueness in the case that the second-order operator is uniformly elliptic and the jump part of $L$ is a compound Poisson process.
\begin{assumption}\label{pham}
Assume in \eqref{L} that $\sigma(t)>0$ for all $t\in[0,T]$ and $\nu$ is a finite L\'{e}vy measure.
\end{assumption}
\begin{prop}[Finite Lévy measure]\label{41}
Under Assumption \ref{pham}, the function
\begin{equation*}
G(t,s):=\bE\left[\int_t^T f^*(u,S^{t,s}(u)) \,du\right]
\end{equation*}
is the unique $C_{\nu,\psi}^{1,2}([0,T)\times\bR)\cap C([0,T]\times\bR)$ solution of \eqref{97}. Moreover, the following integrability conditions hold:
\begin{align*}
&\bE\left[\int_{t}^T \int_\bR \left[G(u,S^{t,s}(u-)+\psi(u) y)-G(u,S^{t,s}(u-))\right]^2\,\nu(dy) \, du\right]<\infty,\\
&\bE\left[\int_{t}^T \sigma(u)^2\, G_s(u,S^{t,s}(u))^2\, du\right]<\infty.
\end{align*}

\end{prop}
\begin{proof}
First of all, observe that, since $\nu$ is the Lévy measure associated to a compound Poisson process, the spaces $C_{\nu,\psi}^{1,2}([0,T)\times\bR)$ and $C^{1,2}([0,T)\times\bR)$ coincide (cf. \cite[Definition 17.4.9]{MR3443368}). Therefore, we only need to verify if the assumptions of \cite[Proposition 5.3]{MR1650147} are fulfilled. Notice that (H6) there corresponds to assuming that the Lévy jump component is a compound Poisson process. Then, in order to apply \cite[Proposition 5.3]{MR1650147} it remains to prove that $f^*:[0,T]\times\bR\to\bR$ is Lipschitz continuous, which follows from Proposition \ref{91}. 
The integrability conditions can be proved as in Lemma \ref{c}.
\end{proof}

Finally, we apply the Verification Theorem and state the main results of this section.

\begin{theorem}\label{36}
Let $g$ be a $C_{\nu,\psi}^{1,2}([0,T)\times\bR)\cap C([0,T]\times\bR)$ solution of \eqref{97} and assume that, for any $t\in[0,T)$ and $s\in\bR$, we have the following conditions
\begin{align*}
&\bE\left[\int_{t}^T \int_\bR \left[g(u,S^{t,s}(u-)+\psi(u) y)-g(u,S^{t,s}(u-))\right]^2\,\nu(dy) \, du\right]<\infty,\\
&\bE\left[\int_{t}^T \sigma(u)^2\, g_s(u,S^{t,s}(u))^2\, du\right]<\infty,
\end{align*}
where $S=S^{t,s}$ is the solution of
\begin{align*}
    dS(u)  &= -\lambda S(u)\,du + dL(u),& u\in(t,T],  \\ 
\label{8}
    S(t) &= s. &
\end{align*}
Then, the function $\pi^*(t,s,x):=\bar\pi^*(t,s)\cdot x$, with $\bar\pi^*$ as in Proposition \ref{19}, is an optimal Markov control policy, i.e. it induces an admissible strategy in the sense of Definition \ref{11} and, for each $t\in[0,T), s\in\bR, x\in\bR^+$, we get that $J(t,s,x;\pi^*)=V(t,s,x)=\log(x)+g(t,s)$.
\end{theorem}
\begin{proof}
See Appendix C.
\end{proof}
\begin{corollary}
Assume either Assumption \ref{pure-jump}, or Assumption \ref{pham} and define
$$
G(t,s)=\bE\left[\int_t^T f^*(u,S^{t,s}(u)) \,du\right].
$$
Then, $\pi^*(t,s,x):=\bar\pi^*(t,s)\cdot x$, as in Proposition \ref{19}, is an optimal Markov control policy and $J(t,s,x;\pi^*)=V(t,s,x)=\log(x)+G(t,s)$.
\end{corollary}

\section{Estimating the Optimal Strategy: the Merton Ratio and Taylor Approximations}

After proving the existence and describing the analytical properties of the optimal strategy, we now study simple ways to compute it by approximation.
 
\subsection{Definition and intuition} In his seminal work on portfolio selection \cite{merton}, Merton studies the optimal allocation of the investor's wealth when the risky asset follows a geometric Brownian motion:
\begin{align*}
dS(t) =\mu S(t)\,dt + \sigma S(t)\,dW(t),  \qquad t\in(0,T],
\end{align*}
and finds that the optimal allocation for a log-utility is\footnote{Consistently with our setting, we are assuming that the risk-free interest rate $r$ is zero and there is no consumption during the trading period.}
$$
\pi^*_M = \frac\mu{\sigma^2},
$$
which consists of the ratio of the excess return over the local variance of the log-price. Instead, in our framework the price dynamics are
\begin{align*}
    dS(t) &= (b(t)-\lambda S(t))\,dt + \sigma(t)\,dW(t) + \psi(t) \int_\bR y\,\overline N(dy,dt), \qquad t\in(0,T].
\end{align*}
Here, the local variance at time $t$ is the sum of the variance of the continuous component $\sigma(t)^2$ and that of the jump part $\sigma_L(t)^2:=\psi(t)^2 \int_\bR y^2\, \nu(dy)$. Then, in this context, it is natural to define the analogue of Merton's Ratio $\pi^*_M$ as
\begin{equation}\label{15}
\bar\pi_1^*(t,s):=\frac{b(t)-\lambda s}{\sigma(t)^2+\sigma_L(t)^2}.
\end{equation}
This ratio appears naturally when applying a Taylor approximation to \eqref{7}. Recall that the optimal normalized strategy $\bar\pi^*$ is defined as 
$$
\bar\pi^*(t,s)=\arg\max_{\bar\pi\in\Pi} f(\bar\pi;t,s),
$$
where
$$
f(\bar\pi;t,s)=(b(t)-\lambda s) \bar\pi-\frac12\,\sigma(t)^2 \bar\pi^2+ \int_\bR 
\left[\log\left(1+\bar\pi \psi(t) y\right)-\bar\pi \psi(t) y\right]
\, \nu(dy).
$$
If the maximum is attained at an interior point (cf. Proposition \ref{19}), then $\bar\pi^*$ satisfies the integral equation 
\begin{equation}\label{95}
b(t)-\lambda s -\sigma(t)^2\bar\pi-\int_\bR \frac{\bar\pi \psi(t)^2 y^2}{1+\bar\pi \psi(t) y}\, \nu(dy)=0.
\end{equation}
If we replace the integrand by the second-order Taylor expansion around zero, the integral equation becomes
\begin{equation}\label{13}
b(t)-\lambda s -\sigma(t)^2 \bar\pi-\bar\pi \psi(t)^2 \int_\bR  y^2\, \nu(dy)=0,
\end{equation}
whose unique solution is exactly the strategy {\it \`{a} la} Merton $\bar\pi_1^*(t,s)$ that we defined in \eqref{15}. A similar Taylor truncation has been introduced in \cite{MR3176490} to study approximations of L\'{e}vy processes and tested numerically in \cite{pasin}. Also, \cite{ascheberg,imma_reno} arrive at an analogous approximated strategy for stochastic volatility models with jumps, however they start by approximating the HJB directly. The idea is to treat the small jumps as an additional Brownian component (very much in the spirit of \cite{asm_ros}) and neglect larger jumps. 

Let us now introduce a more accurate approximation in the finite activity case. We assume that $\nu([m,M]\setminus\{0\})<\infty$, i.e. the jump component of \eqref{5} is, in fact, a compound Poisson process, which is often the most interesting case for application purposes. Therefore, the L\'{e}vy measure takes the form $\nu(dy)=\eta\,F(dy)$, where $\eta$ is the jump intensity and $F(dy)$ the jump size distribution. Thanks to our standing assumptions, the distribution $F$ admits finite expectation $\mu_F$ and variance $\sigma^2_F$. 
We remind that the optimal policy is defined through the first order condition in \eqref{95}. 
Let us write the Taylor polynomial of the integrand around an arbitrary (finite) point $y_0\in[m,M]\setminus\{0\}$. So, we set $\phi(y):=\frac{\bar\pi y^2}{1+\bar\pi \psi(t) y}$ and compute its derivatives. Writing down its expansion up to the first order, we have
\begin{equation}\label{40}
\phi(y)=\frac{\bar\pi y_0^2}{1+\bar\pi \psi(t) y_0}+\frac{2\bar\pi y_0+\bar\pi^2 y_0^2 \psi(t)}{(1+\bar\pi \psi(t) y_0)^2}(y-y_0)+o(y-y_0).
\end{equation}
Then, \eqref{95} becomes
\begin{equation}\label{94}
b(t)-\lambda s -\sigma(t)^2\bar\pi-\frac{\bar\pi \psi(t)^2 y_0^2}{1+\bar\pi \psi(t) y_0} \nu([m,M]\setminus\{0\})
-\frac{2\bar\pi \psi(t)^2 y_0+\bar\pi^2 \psi(t)^3 y_0^2}{(1+\bar\pi \psi(t) y_0)^2}\int_\bR (y-y_0)\, \nu(dy)=0.
\end{equation}
From this expression it is clear that a significant simplification is given by the choice 
\begin{equation*}
y_0:=\frac1{\nu([m,M]\setminus\{0\})} \int_\bR y\, \nu(dy),
\end{equation*}
since in this case the first-order term just disappears. Besides, by writing the L\'{e}vy measure with respect to $F$, we see that $y_0$ corresponds to the mean value of the jump size:
\begin{equation}\label{0}
y_0=\frac1{\eta\,\int_\bR F(dy)} \eta\,\int_\bR y\, F(dy)=\mu_F.
\end{equation}
We are essentially replacing the integrand with its linear approximation around the integral mean, or, from another point of view, we are approximating a function of the jumps with respect to the jump size mean value. Since here we take into account the jump measure specifications, this is a slightly different approach from the first approximation $\bar\pi_1^*$ (where the Taylor polynomial was centered in zero). Hence, from \eqref{94} we have the following approximated equation for $\bar\pi$:
\begin{equation}\label{26}\begin{split}
-\sigma(t)^2 \psi(t) \mu_F \bar\pi^2+(\mu_F \psi(t)(b(t)-\lambda s)-\sigma(t)^2-\mu_F^2 \psi(t)^2 \eta)\bar\pi+b(t)-\lambda s=0.
\end{split}\end{equation}
In the case $\sigma(t) \psi(t)\not=0$, it is a second order polynomial in the variable $\bar\pi$. Consequently, for each $t\in[0,T]$ we have two (generally) different solutions. However, only one of them is admissible, meaning that $\bar\pi_2^*(t,s)\in\Pi$ for every possible value of $s$. Since the ambiguity comes from the definition domain of the logarithm, we just need to impose the condition $1+\bar\pi^*_2 \psi(t) \mu_F>0$. This leads to 
$$
\bar\pi_2^*(t,s):=
\begin{cases}
\frac{p_1(t,s)+\sqrt{p_1(t,s)^2+4 p_2(t,s)} }{2 p_3(t,s)}, & \mbox{if } \mu_F>0,\\
\frac{p_1(t,s)-\sqrt{p_1(t,s)^2+4 p_2(t,s)} }{2 p_3(t,s)}, & \mbox{if } \mu_F<0,
\end{cases}
$$
where 
\begin{align*}
p_1(t,s)&=\mu_F \psi(t) (b(t)-\lambda s)-\mu_F^2 \psi(t)^2 \eta-\sigma(t)^2,\\ p_2(t,s)&= \mu_F(b(t)-\lambda s)\sigma(t)^2 \psi(t), \\
p_3(t,s)&= \mu_F \psi(t) \sigma^2. \\
\end{align*}
On the other hand, in a pure jump context ($\sigma(t)\equiv0$), we get simply
\begin{equation}\label{32}
\bar\pi_2^*(t,s)=-\frac{b(t)-\lambda s}{\psi(t)\mu_F \bigl(b(t)-\lambda s-\eta \psi(t)\mu_F\bigr)},
\end{equation}
for any $\mu_F\not=0$ and for each $t$ and $s$ such that $\bar\pi_2^*(t,s)$ is well defined and takes values into $\Pi$. 

\subsection{Error bounds} To estimate the approximation error, we compute the difference between the optimal normalized strategy $\bar\pi^*$ and the approximated ones $\bar\pi^*_1$ and $\bar\pi^*_2$. 
\begin{prop}\label{34}
Assume that $\Pi$ contains $0$ and the finiteness of the third moment of the jumps, that is 
$$
\int_{{|y|\geq1}} |y|^3\, \nu(dy)<\infty,
$$
and denote $\sigma_1^2= \min_{[0,T]}\sigma^2(t),\sigma_\nu^2=\int_\bR y^2\, \nu(dy)$. 
If we are not in {\bf case D} (no jumps) of Remark \ref{18}, then for each $t\in[0,T]$ it holds that
$$
|\bar\pi^*(t,\cdot)-\bar\pi^*_1(t,\cdot)| \leq  C \int_\bR |y|^3\, \nu(dy),
$$
where $C$ is a constant which depends on $\delta:=${\rm dist}$(\Pi,\partial\widehat\Pi),$ $\max\Pi,$ $\min\Pi,$ $m,$ $M,$ $\psi_2$, $\sigma_1^2$ and $\sigma_\nu^2$ according to the following different cases:
\begin{description}
\item[case A] $$
C=\begin{cases}\frac{C_0}{\min\{1,\delta \psi_2 M,-\delta \psi_2 m\}}
& \mbox{if } m\not=-\infty, M\not=+\infty,\\
\frac{C_0}{\min\{1,\delta \psi_2 M\}}
& \mbox{if } m=-\infty, M\not=+\infty,\\
\frac{C_0}{\min\{1,-\delta \psi_2 m\}}
& \mbox{if } m\not=-\infty, M=+\infty.
\end{cases}
$$
\item[case B]
$$
C=\begin{cases}
\frac{C_0}{\min\{1,\delta \psi_2 M\}}
& \mbox{if } M\not=+\infty,\\
C_0 & \mbox{if } M=+\infty.
\end{cases}
$$
\item[case C] $$
C=\begin{cases}
\frac{C_0}{\min\{1,-\delta \psi_2 m\}}
& \mbox{if } m\not=-\infty,\\
C_0 & \mbox{if } m=-\infty,
\end{cases}
$$
\end{description}
where $C_0:=\frac{\psi_2^3}{\sigma_1^2+\psi_2^2 \sigma_\nu^2} \max_{\pi\in\Pi}\pi^2 $.
\end{prop}

\begin{prop}\label{99}
Let us assume that we are not in {\bf case D} of Remark \ref{18}. Moreover, we suppose $0\in\Pi$ and that $\sigma(t)^2$ and $\mu_F$ are not both identically 0. Then for each $t\in[0,T]$ it holds that
$$
|\bar\pi^*(t,\cdot)-\bar\pi^*_2(t,\cdot)| \leq \frac {\eta \psi_2^2\,C_1\,\sigma_F}{\sigma_1^2+ \eta\,\psi_2^2\, C_2\,\mu_F^2},
$$
where $\sigma_1^2= \min_{[0,T]}\sigma^2(t)$, $\sigma_F$ is the square root of the variance of the random jump size and $C_1,C_2$ are constants depending on $\delta:=\mbox{dist}(\Pi,\partial\widehat\Pi),$ $\max\Pi,$ $\min\Pi,$ $m,$ $M,$ $\psi_2$ according to the following different cases:
\begin{description}
\item[case A] $$
C_1=
\begin{cases}\max\{1,\frac{1}{(\delta \psi_2 M)^2},\frac{1}{(\delta \psi_2 m)^2}\}+\max\{1,\frac{1}{\delta \psi_2 M},\frac{1}{-\delta \psi_2 m}\}, & \mbox{if } m\not=-\infty, M\not=+\infty,\\
\max\{1,\frac{1}{(\delta \psi_2 M)^2}\}+\max\{1,\frac{1}{\delta \psi_2 M}\}, & \mbox{if } m=-\infty, M\not=+\infty,\\
\max\{1,\frac{1}{(\delta \psi_2 m)^2}\}+\max\{1,\frac{1}{-\delta \psi_2 m}\}, & \mbox{if } m\not=-\infty, M=+\infty.
\end{cases}
$$
$$
C_2=\begin{cases}
\frac{1}{(1+\max\Pi \,\psi_2\,\mu_F)^2} & \mbox{if } \mu_F>0,\\
\frac{1}{(1+\min\Pi \,\psi_1\,\mu_F)^2} & \mbox{if } \mu_F<0.
\end{cases}
$$
\item[case B]
$$
C_1=\begin{cases}
\max\{1,\frac{1}{(\delta \psi_2 M)^2}\}+\max\{1,\frac{1}{\delta \psi_2 M}\}, & \mbox{if } M\not=+\infty,\\
1 & \mbox{if } M=+\infty.
\end{cases}
$$
$$
C_2=\frac{1}{(1+\max\Pi \,\psi_2\,\mu_F)^2}.
$$
\item[case C] $$
C_1=\begin{cases}
\max\{1,\frac{1}{(\delta \psi_2 m)^2}\}+\max\{1,\frac{1}{-\delta \psi_2 m}\}, & \mbox{if } m\not=-\infty,\\
1 & \mbox{if } m=-\infty.
\end{cases}
$$
$$
C_2=\frac{1}{(1+\min\Pi \,\psi_1\,\mu_F)^2}.
$$
\end{description}
\end{prop}

\section{Numerical Example}

In this section we test our trading strategies on one of the most popular electricity price models, namely the factor model in \cite{Benth20121589}. There, the authors conduct a critical comparison of three different spot price models for electricity in the context of day-ahead markets. In fact, this is typically an auction market in which the electricity price is fixed for the subsequent day, so that daily averaged prices are taken into account, over a timeline of years. Consequently, our setting concerns different market and price definitions. To recall, we instead take the point of view of an agent in the intraday market, which is the exchange where the electricity is traded continuously for 8-27 hours (depending on the contract) generally in the form of quarterly or hourly forward contracts. However, we aim to exploit the analysis in \cite{Benth20121589}, where the model is calibrated to Nord Pool Spot market data, for mainly two reasons. Firstly, the stylized features of intraday markets are of similar nature as the ones observed in the day-ahead price series, such as spike behavior, high volatility, leptokurtosis (for a more detailed empirical study see, for instance, \cite{gallana,kieselpara}).  Secondly, the factor model is based on a L\'{e}vy Ornstein-Uhlenbeck process of the same family as the one in \eqref{5}. 

We consider the factor model as in \cite{Benth20121589}, which was originally introduced for electricity price modelling in \cite{MR2323278}. The price dynamics are written as
$$
S(t)=e^{Q(t)} Z(t),
$$
where 
$$
Z(t)=\sum_{i=1}^n w_i\,Y_i(t)
$$
is the deseasonalized price and $Q(t)$
is the seasonal component. 
The $w_i$ are positive weights while the factors $Y_i$ $(i=1,...,n)$ are independent non-Gaussian Ornstein-Uhlenbeck processes described by
$$
dY_i(t)=-\lambda_i\,Y_i(t)\,dt + dL_i(t),\qquad Y_i(0)=y_i,\qquad i=1,...,n,
$$
being $L_i$ independent \emph{c\`{a}dl\`{a}g} pure jump additive processes with increasing paths.

As the authors calibrate the model in \cite{Benth20121589}, by comparing the theoretical autocorrelation function to the empirical one, they set the optimal number of factors to $n=2$. The estimated speeds of mean reversion are $\lambda_1=0.0087$ and $\lambda_2=0.3333$. In the paper these two values are interpreted as, respectively, the \emph{base} (slowest) and the \emph{spike} (fastest) signal. 
We start from here to define our equations. Specifically, their data series ranges from 13/07/2000 to 7/08/2008, which comprises, excluding the weekends, $2099$ days. The time unit for $t$ is 1 day. So, in order to adapt it to our timeline, which covers hours of intraday transactions, first we set one hour as our time unit, that is we do the time variable change $u=24\cdot t$ . 
Then, denoting $C:=24$, we set our mean reversion speed $\lambda$ in our own model by rescaling in time the spike speed ($\lambda_2=0.3333$), i.e. we take $\lambda=\lambda_2/C=0.0139$.
The driving process $L$ is a compound Poisson process where the jump intensity, originally adopted in the paper by \cite{gemanR}, is seasonally dependent. Also, the jump size distribution is a Pareto($\alpha$,$z_0$), with $\alpha=2.5406$, $z_0=0.3648$ and density function $f(y)=\frac{\alpha z_0^\alpha}{y^{\alpha+1}}$. Therefore, we are in the case of positive jumps: with the notation of previous sections, $\mbox{supp} (\nu)=[z_0,+\infty)$, i.e. $m=z_0$, $M=+\infty$, $\widehat\Pi=[0,+\infty)$ and $F(dy)=f(y)dy$. The set $\Pi$ can be any compact subset of $\widehat\Pi$ containing $0$.

For our own problem to make sense, another issue to address is to deseasonalize the jump intensity. In details, the form of the intensity is the following
$$
e(t)=\theta\cdot s(t)=\theta\cdot \left( \frac{2}{1+|\sin(\pi\frac{t-\tau}{k})|}-1\right)^d
$$ 
where $\theta=14.0163$ represents the expected number of spikes per time unit at a spike-clustering time, whereas the seasonal parameters are set by the authors' calibration procedure $k=0.5$, $\tau=0.42$, $d=1.0359$. We then decide to compute the integral mean of $e(t)$ over its time periodicity, that is $2k$, obtaining $\mu:=3.7249$, so that, after rescaling, we have our intensity $\eta:=\mu/C=0.1552$. 

To summarize, the electricity price in our hourly intraday market is described by

\begin{equation}\label{20}
dS(t)  = (b(t)-\lambda S(t))\,dt + dL(t),
\end{equation}
where $\lambda=\lambda_2/C=0.0139$ is the mean reversion speed and $L$ is a compound Poisson process with jump intensity $\eta=\mu/C=0.1552$ and jump size distribution a Pareto law of parameters $\alpha=2.5406$ and $z_0=0.3648$. Therefore, there is no Brownian component in the jumps and the coefficient of the jump volatility is normalized to 1 ($\sigma^2\equiv0$ and $\psi\equiv1$). In particular it is important to notice that $L$ is a subordinator, which keeps the price positive. The drift value $b$ cannot be derived directly from \cite{Benth20121589}, being not part of the spike signal and for this reason it will be discussed later. 

Let us write the equation for the exact normalized strategy $\bar\pi^*=\bar\pi^*(t,s)$, defined in the integral equation \eqref{7}, i.e.
$$
(b(t)-\lambda s) -\eta\,\int_\bR \frac{\bar\pi^* y^2}{1+\bar\pi^* y}\, f(y)\,dy=0,
$$
or, expressing it in terms of the price level $s$,
\begin{equation}\label{16}
s=\frac{b(t)}{\lambda}-\frac{\eta}{\lambda}\cdot\int_{z_0}^\infty \frac{\bar\pi^* y^2}{1+\bar\pi^* y}\,f(y)\,dy,
\end{equation}
where $f(y)=\frac{\alpha z_0^\alpha}{y^{\alpha+1}}$ is the density of a Pareto law and $\eta$ the jump intensity of $L$.
By using a software to integrate exactly the above expression (we used \emph{Mathematica}\texttrademark) and inserting the estimated parameters inside the integral, we get 
$$
s=\frac{b(t)}{0.0139}-6.7233 \cdot\, _2F_1\left(1,1.5406;2.5406;-\frac{2.7412}{\bar\pi^*}\right),
$$
where $_2F_1(a,b;c;z)$ is the hypergeometric function. This explicit formula states the value of the price $s$ with respect to the optimal strategy $\bar\pi^*$. The inverse relation $s\mapsto \bar\pi^*$ can be computed numerically.
As already observed, if we are interested in plotting the functions above, we need to set a value for the drift $b(t)$, being not consistently computable from the analysis in \cite{Benth20121589}. A quick study of the last expression yields the following.

\begin{prop}\label{33} 
Let us recall the definition of the jump measure mean
$$
\mu_L=\int_\bR y\, \nu(dy)=\eta\,\int_\bR y\, f(y)\,dy=\eta\,\mu_F,
$$
where $\mu_F$ is the mean of the jump size distribution $f(y)dy$. For each $t\in[0,T]$, if the drift $b(t)$ in the price equation \eqref{20} is nonpositive, then $\bar\pi^*(t,s)\equiv 0$ for any $s\in\bR^+$. Furthermore, if the drift $b(t)$ is greater than or equal to $\mu_L$ for $t\in[0,T]$, then $\bar\pi^*(t,s)\equiv \max\Pi$ for any $s\leq s^*(t,\max\Pi)$, where $s^*(t,\cdot)$ denotes the inverse function of $\bar\pi^*(t,\cdot)$ (within its range of invertibility). 
\end{prop}
\begin{proof}
On one hand, the admissible values of the strategy belong to $\Pi$, which is any compact subset of $\widehat\Pi=[0,+\infty)$ containing 0. On the other hand, the price $s$ can take only positive values by construction (recall that $L$ is a subordinator). In \eqref{16} the function $\bar\pi^*\mapsto \eta\,\int_{z_0}^\infty \frac{\bar\pi^* y^2}{1+\bar\pi^* y}\,f(y)\,dy$ is increasing on the positive real line and it holds
\begin{align*}
\lim_{\bar\pi^*\to 0} \eta\,\int_{z_0}^\infty \frac{\bar\pi^* y^2}{1+\bar\pi^* y}\,f(y)\,dy&=0,
\\
\lim_{\bar\pi^*\to +\infty}\eta\,\int_{z_0}^\infty \frac{\bar\pi^* y^2}{1+\bar\pi^* y}\,f(y)\,dy&=\mu_L>0.
\end{align*}
Therefore, the first order condition in \eqref{16} is not satisfied for admissible values of $s$ and $\bar\pi^*$ whenever $b(t)\leq0$ since
$$
s=\frac 1\lambda \left(b(t)-\eta\cdot\int_{z_0}^\infty \frac{\bar\pi^* y^2}{1+\bar\pi^* y}\,f(y)\,dy\right)<0.
$$
This means that the maximum in \eqref{24} is attained at the boundary of $\Pi$ and more exactly when $\bar\pi^*=\min\Pi=0$ (cf. Proposition \ref{19}). The case $b(t)\geq\mu_L$ can be proven along exactly the same reasonings.
\end{proof}
Now, we write from \eqref{15} the first approximated strategy:
$$
\bar\pi_1^*(t,s)=
\begin{cases}
\frac{b(t)-\lambda s}{\sigma_L^2},&\mbox{if }\bar\pi_1^*(t,s)\in\Pi,\\
0,&\mbox{if }\bar\pi_1^*(t,s)\not\in\Pi,
\end{cases}
$$
A straightforward computation yields $\sigma^2_L=\int_{m}^M y^2\,\nu(dy)=\eta\cdot\int_{z_0}^\infty y^2\,f(y)\,dy=\eta\cdot 0.6254=0.0971$. Observe that, in order that condition $\bar\pi_1^*(t,s)\in\Pi$ holds, $b(t)$ must be non-negative. Having chosen any $\Pi$  compact subset of $\widehat\Pi=[0,+\infty)$ containing $[0,\frac{b(t)}{\sigma^2_L}]$ for any $t\in[0,T]$, this reads
$$
\bar\pi_1^*(t,s)=
\begin{cases}
\frac{\lambda}{\sigma^2_L}\cdot\left(\frac{b(t)}{\lambda}-s\right), 
&\mbox{if }0\leq s\leq\frac{b(t)}{\lambda},
\\
0, &\mbox{if }s>\frac{b(t)}{\lambda}.
\end{cases}
$$
Finally, from \eqref{32} we get the second approximation of the optimal strategy:
$$
\bar\pi_2^*(t,s)=
\begin{cases}
-\frac{\eta}{\mu_L}\,\frac{b(t)-\lambda s}{b(t)-\lambda s-\mu_L},
& \mbox{if }0\leq s\leq\frac{b(t)}{\lambda},
\\
0, &\mbox{if }s>\frac{b(t)}{\lambda},
\end{cases}
$$
still with the condition that $\bar\pi_2^*(t,s)\in\Pi$ (cf. Proposition \ref{33}). 

\begin{figure} 
\begin{subfigure}{0.48\textwidth}
\includegraphics[width=\linewidth]{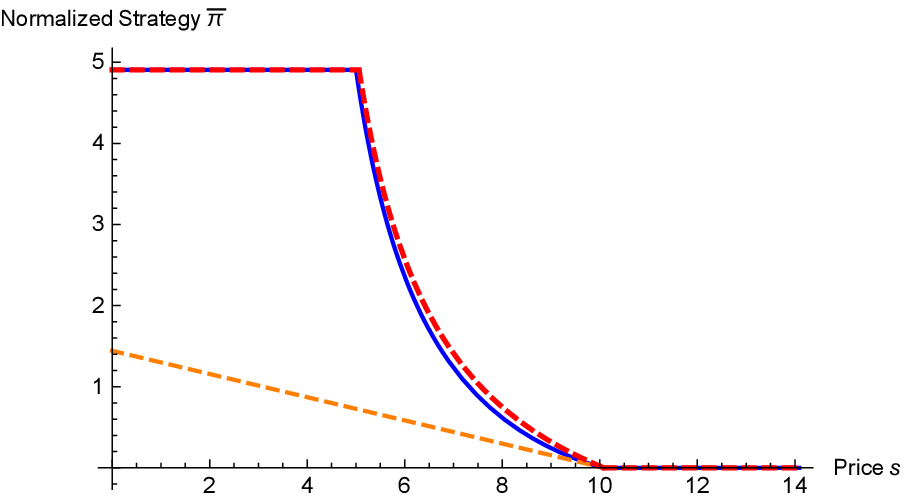}
\caption{$b=150\%$ of $\mu_L$. 
} \label{fig:a}
\end{subfigure}\hspace*{\fill}
\begin{subfigure}{0.48\textwidth}
\includegraphics[width=\linewidth]{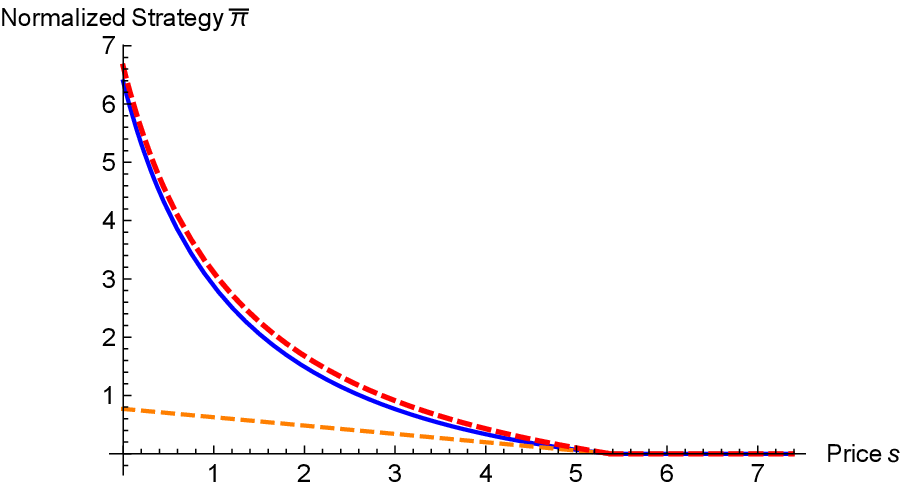}
\caption{$b=80\%$ of $\mu_L$. 
} \label{fig:b}
\end{subfigure}

\vspace{30pt}
\medskip

\begin{subfigure}{0.48\textwidth}
\includegraphics[width=\linewidth]{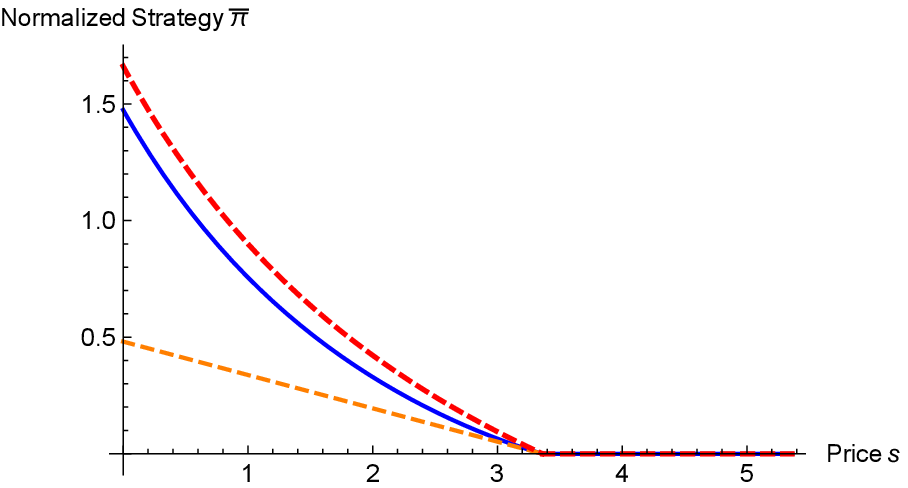}
\caption{$b=50\%$ of $\mu_L$. 
} \label{fig:c}
\end{subfigure}\hspace*{\fill}
\begin{subfigure}{0.48\textwidth}
\includegraphics[width=\linewidth]{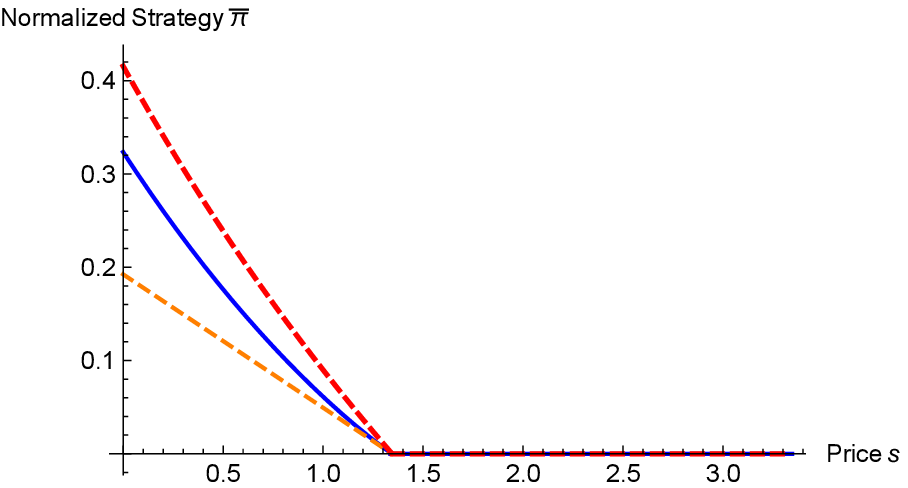}
\caption{$b=20\%$ of $\mu_L$. 
} \label{fig:d}
\end{subfigure}

\vspace{30pt}

\caption{Blue line: (exact value) $\bar\pi^*$, orange line: $\bar\pi_1^*$, red line: $\bar\pi_2^*$.} \label{fig:1}
\end{figure}

We now plot the different strategies: the exact and the two approximations. In view of Proposition \ref{33} we do it for the following values of the drift $b=150,80, 50,20\%$ of $\mu_L$. In such a way we can understand their behavior in the most representative cases (see Figure \ref{fig:1}). Remind that if $s\geq\frac b \lambda$, then $\bar\pi^*$, $\bar\pi_1^*$ and $\bar\pi_2^*$ are all identically equal to 0. We observe from our numerical results the following facts:
\begin{enumerate}
\item The order among the strategies: $ \bar\pi_1^*\leq\bar\pi^*\leq\bar\pi_2^*$ holds.
\item As $b$ approaches (and exceeds) $\mu_L$, the second approximation $\bar\pi_2^*$ gets much better, until it becomes almost indistinguishable from the exact strategy, while if $b$ approaches $0$, the (considerable) error between the first approximation $\bar\pi_1^*$ and the exact value decreases. In both cases, the shapes of the two approximations are similar to the one of $\bar\pi^*$. For instance, in the latter case the optimal strategy flattens out and looks like a straight line.
\item The bad performance of $\bar\pi_1^*$ may be explained from the fact that it does not satisfy the requirements for the estimate in Proposition \ref{34}. This happens because the Pareto law estimated by \cite{Benth20121589} has parameter $\alpha=2.5406<3$, which means that it admits finite second moment but not finite third moment (see assumptions of Proposition \ref{34}). Moreover, this approximation is natural for processes with small jumps, whereas the second one, i.e. $\bar\pi_2^*$, is more consistent with general jump processes since it is constructed around the jump measure mean $\mu_L$. What is particularly interesting is that an economically meaningful quantity as the Merton Ratio, that we translated into $\bar\pi_1^*$ (see Equation \ref{15}), performs generally much worse than the Taylor approximation.
\item As we already mentioned, essentially the same approximation $\bar\pi_1^*$ is numerically investigated in \cite{pasin}. What the authors found there is that it works rather well for three popular price models. The difference from our setting, which could even explain why we observe such an unsatisfactory performance, is that they are in the context of exponentially additive models, while our price dynamics are purely additive and mean-reverting.
\end{enumerate}

\appendix

\section{}

The following lemmas are auxiliary results for Proposition \ref{z}. Let us recall that $S=S^{t,s}$ is described by
\begin{align*}
dS(u)  &= (b(u)-\lambda S(u))\,du + \psi(u)\,\int_\bR y\,\overline N(dy,du), \qquad u\in(t,T],\\
S(t)&=s,
\end{align*}
and can be written explicitly as
$$
S^{t,s}(u)=s e^{-\lambda (u-t)}+\int_t^u e^{-\lambda (u-v)}b(v)\, dv+\int_t^u \int_\bR e^{-\lambda (u-v)}\psi(v) y\,\overline N(dy,dv). 
$$
Furthermore, the candidate solution for the PIDE in Proposition \ref{z} is defined as 
\begin{equation}\label{G}
G(t,s)=\bE\left[\int_t^T f^*(u,S^{t,s}(u)) \,du\right].
\end{equation}

\begin{lemma}\label{a}
For all $t\in[0,T]$ and $s\in\bR$, it holds that 
$$
\bE\left[\int_t^T |S^{t,s}(u)| \,du\right]<\infty.
$$
\end{lemma}
\begin{proof}
For $u\in[t,T]$, we have
\begin{align*}
\bE\left[|S^{t,s}(u)|\right]&\leq |s| e^{-\lambda (u-t)}
+\left|\int_t^u e^{-\lambda (u-v)}b(v)\, dv\right| 
&+\bE\left[\left|\int_t^u \int_\bR e^{-\lambda (u-v)} \psi(v)  y \, \overline N(dy,dv)\right|\right] 
\end{align*}
Since
\begin{align*}
\left(\bE\left[\left|\int_t^u \int_\bR e^{-\lambda (u-v)} \psi(v)  y \, \overline N(dy,dv)\right|\right]\right)^2
&\leq \bE\left[\left(\int_t^u \int_\bR e^{-\lambda (u-v)} \psi(v)  y \, \overline N(dy,dv)\right)^2\right] \\
&=\int_t^u \int_\bR e^{-2\lambda (u-v)} \psi(v)^2  y^2 \, \nu(dy)dv \\
&\leq \psi_2^2 \left( \int_\bR  y^2 \, \nu(dy)\right) 
\left(\frac{1- e^{-2\lambda (u-t)}}{2\lambda}\right)
\end{align*}
and 
$$
\left|\int_t^u e^{-\lambda (u-v)}b(v)\, dv\right| \leq
\int_t^u e^{-\lambda (u-v)}|b(v)|\, dv\leq C \left(\frac{1- e^{-\lambda (u-t)}}{\lambda}\right)
$$
we find that 
$$
\int_t^T \bE\left[|S^{t,s}(u)|\right] \,du<\infty.
$$
We conclude by Tonelli's Theorem.
\end{proof}

\begin{lemma}\label{b}
The function $G:[0,T]\times\bR\to\bR$ in Equation \eqref{G} is well defined. In particular, 
$$
\bE\left[\int_t^T |f^*(u,S^{t,s}(u))| \,du\right]<\infty.
$$
\end{lemma}
\begin{proof}
Just observe that 
$$
|G(t,s)|\leq \bE\left[\int_t^T |f^*(u,S^{t,s}(u))| \,du\right]\leq C \left(1+\bE\left[\int_t^T |S^{t,s}(u)| \,du\right]\right),
$$
which is finite by Lemma \ref{a}.
\end{proof}

\begin{lemma}\label{c}
The function $G:[0,T]\times\bR\to\bR$ defined in \eqref{G} is continuous in the time variable for any fixed $s\in\bR$ and continuously differentiable  in $s$ for any fixed time $t\in[0,T)$ with bounded derivative. Specifically, 
$$
\partial_s G(t,s) = \bE\left[-\int_t^T \lambda e^{-\lambda (u-t)} \bar\pi^*(u,S^{t,s}(u)) \,du\right].
$$
Furthermore, $\partial_s G(t,s)$ is Lipschitz continuous in the variable $s$ uniformly in $t\in[0,T]$.
\end{lemma}
\begin{proof}
For the continuity observe that 
\begin{align*}
G(t+h,s)-G(t,s)=\bE\left[\int_{t+h}^T f^*(u,S^{t+h,s}(u)) \,du\right]
-\bE\left[\int_{t}^T f^*(u,S^{t,s}(u))\,du\right] \\
=\bE\left[\int_{t+h}^T [f^*(u,S^{t+h,s}(u))-f^*(u,S^{t,s}(u))]\,du \right]
-\bE\left[\int_{t}^{t+h} f^*(u,S^{t,s}(u))\,du\right].
\end{align*}
As $h$ tends to zero, the second term vanishes by the dominated convergence theorem (cf. Lemma \ref{b}). For the first term observe that
\begin{align*}
\int_{t+h}^T \bE\left[|f^*(u,S^{t+h,s}(u))-f^*(u,S^{t,s}(u))| \right]\,du 
\leq L \int_{t+h}^T \bE\left[|S^{t+h,s}(u)-S^{t,s}(u)| \right]\,du.
\end{align*}
Hence,
\begin{align*}
\bE[|S^{t+h,s}(u)-S^{t,s}(u)|]&\leq \bE[(S^{t+h,s}(u)-S^{t,s}(u))^2]\\
&\leq 3 s^2 e^{-2\lambda (u-t)}(e^{\lambda h}-1)^2
+3\left(\int_t^{t+h} e^{-\lambda (u-v)}b(v)\, dv\right)^2 \\
&+3 \left( \int_\bR y^2 \, \nu(dy)\right)\left(\int_t^{t+h}  e^{-2\lambda (u-v)} \psi(v)^2 dv \right),
\end{align*}
which is Lebesgue-integrable for $u\in[t,T]$ and approaches zero as $h$ tends to zero. Then, by the dominated convergence theorem, it holds that
$$
\int_{t+h}^T \bE\left[|S^{t+h,s}(u)-S^{t,s}(u)| \right]\,du
$$
vanishes and that $G(\cdot,s)$ is continuous.

In order to prove differentiability, we apply the classical theorem about differentiation under the integral sign. First, define 
$$
F(t,s):=\int_t^T f^*(u,S^{t,s}(u)) \,du.
$$
Since $f^*(u,S^{t,s}(u))$ is continuously differentiable for each $u$ with partial derivative dominated by an integrable function:
$$
\partial_s  f^*(u,S^{t,s}(u))=
-\lambda e^{-\lambda (u-t)} \bar\pi^*(t,S^{t,s}(u))\leq C\, e^{-\lambda (u-t)},  
$$
we have
$$
\partial_s F(t,s)=-\lambda\int_t^T e^{-\lambda (u-t)} \bar\pi^*(t,S^{t,s}(u)) \,du.
$$
With the same argument, since 
$$
G(t,s)=\bE\left[F(t,s)\right],
$$
where $F(t,s)$ is differentiable with dominated derivative, we get the statement.

Fix $t\in[0,T]$. Since
$$
\partial_s G(t,s) = \bE\left[-\int_t^T \lambda e^{-\lambda (u-t)} \bar\pi^*(u,S^{t,s}(u)) \,du\right]
$$
and $\bar\pi^*(t,\cdot)$ is uniformly Lipschitz continuous (cf. Proposition \ref{19}), we have
\begin{align*}
|\partial_s G(t,s+h)-\partial_s G(t,s)| =&\left|\bE\left[\int_t^T \lambda e^{-\lambda (u-t)} (\bar\pi^*(u,S^{t,s+h}(u))-\bar\pi^*(u,S^{t,s}(u))) \,du\right]\right|\\
\leq& \bE\left[\int_t^T \lambda e^{-\lambda (u-t)} \left|\bar\pi^*(u,S^{t,s+h}(u))-\bar\pi^*(u,S^{t,s}(u))\right| \,du\right]\\
\leq& L\, \bE\left[\int_t^T \lambda e^{-\lambda (u-t)} \left|S^{t,s+h}(u)-S^{t,s}(u)\right| \,du\right]\\
=& C |h| \left(\int_t^T \lambda e^{-2\lambda (u-t)} \,du\right)\\
=& C |h|,
\end{align*}
where $C$ is a constant depending only on $T$, $\lambda$ and the Lipschitz constant $L$ of $\bar\pi^*(u,\cdot)$ (which is independent of $u$).
\end{proof}

\begin{lemma}\label{d}
For each $s\in\bR$, $t\in[0,T)$ and $h>0$, it holds that 
\begin{align*}
\int_t^{t+h}\int_\bR\bE^{t,s}\left[\bigl(G(t+h,S(u)+\psi(u) y)-G(t+h,S(u))\bigr)^2\right] \nu(dy)du
\end{align*}
and
\begin{align*}
\int_t^{t+h} \int_\bR \bE^{t,s}[|G(t+h,S(u)+\psi(u) y)-G(t+h,S(u))-\partial_s G(t+h,S(u))\psi(u) y|] \nu(dy)du
\end{align*}
are finite.
\end{lemma}
\begin{proof}
For the first term it is sufficient to observe that 
\begin{align*}
\bE^{t,s}\left[\bigl(G(t+h,S(u)+\psi(u) y)-G(t+h,S(u))\bigr)^2\right] \leq \sup_{z\in\bR} G_s(t+h,z)^2 \psi(u)^2 y^2\leq C e^{2\lambda u} y^2. 
\end{align*}
For the second part of the statement, recall from Lemma \ref{c} that $\partial_s G(t,s)$ is Lipschitz continuous in $s$ uniformly in $t$. Let us denote by $\phi(t,s)$ the weak derivative of $\partial_s G(t,s)$ (which is bounded). Therefore, we can write the Taylor expansion of $G(t+h,\cdot)$ in $s+\psi(u) y$ around the center $s$ with integral remainder: 
\begin{align*}
G(t+h,s+\psi(u) y)=G(t+h,s)+\partial_s G(t+h,s)\psi(u) y+\int_s^{s+\psi(u) y} \phi(t+h,\xi) (s+\psi(u) y-\xi) d\xi.
\end{align*}
Hence, for all $s\in\bR$, we have that
\begin{align*}
|G(t+h&,s+\psi(u) y)-G(t+h,s)-\partial_s G(t+h,s)\psi(u) y|\\
&\leq \int_s^{s+\psi(u) y} |\phi(t+h,\xi)| |s+\psi(u) y-\xi| d\xi\\
&\leq C \int_s^{s+\psi(u) y} |s+\psi(u) y-\xi| d\xi = C\, \psi(u)^2 y^2\\
&\leq C\, y^2.
\end{align*}
As a consequence,
\begin{align*}
\int_t^{t+h} \int_\bR \bE^{t,s}[|G(t+h,S(u)+\psi(u) y)-G(t+h,S(u))-\partial_s G(t+h,S(u))\psi(u) y|] \nu(dy)du \\
\leq C\, \int_\bR  y^2 \nu(dy),
\end{align*}
which is finite by our standing assumptions on the Lévy measure $\nu$.
\end{proof}

\section{}

In this Appendix we collect some of the most technical proofs.

\subsection*{Proof of Proposition \ref{19}}

We prove it for {\bf case A}, the other cases being analogous.
Let us denote $\bar\pi_1=\min\Pi$, $\bar\pi_2=\max\Pi$. 
We already observed, due to the concavity of the function $f(\cdot;t,s)$ in \eqref{25}, that the map $\bar\pi^*:[0,T]\times\bR\to[\bar\pi_1,\bar\pi_2]$ is well defined. Also, from \eqref{7} we immediately get that $\bar\pi^*\left(t,\frac{b(t)}\lambda\right)=0$. 

The continuity of the function $\bar\pi^*:[0,T]\times\bR\to\Pi$ relies on a general argument based on the concavity of $f(\bar\pi;t,s)$. If $z=(t,s)\in[0,T]\times\bR$, then $\bar\pi^*=\bar\pi^*(z)=\arg\max_{\bar\pi\in\Pi} f(\bar\pi;z)$. 
Take a sequence $(z_k)_k\subset[0,T]\times\bR$ such that $z_k\to z_0$ as $k\to\infty$. Then the statement follows from proving that $\bar\pi^*(z_k)\to \bar\pi^*(z_0)$. This is equivalent to saying that each subsequence of $\bar\pi^*(z_k)$ admits a subsequence which converges to $\bar\pi^*(z_0)$. Take any subsequence of $\bar\pi^*(z_k)$ and denote it by $\bar\pi^*(z_k)$ (i.e. we do not rename the indexes). Since $\bar\pi^*(z_k)\subset\Pi$, which is compact, it admits a subsequence $\bar\pi^*(z_{k_h})$ converging to a limit $\bar\pi_0\in\Pi$ as $h\to\infty$. Observe that, for any $\bar\pi\in\Pi$,
$$
f(\bar\pi_0;z_0)=\lim_{h\to\infty} f(\bar\pi^*(z_{k_h});z_{k_h})\geq \lim_{h\to\infty} f(\bar\pi;z_{k_h})= f(\bar\pi;z_0).
$$
By definition of $\bar\pi^*$, we have that $\bar\pi_0=\bar\pi^*(z_0)$, which implies that $\bar\pi^*(z_{k_h})\to \bar\pi^*(z_0)$. Since this argument is valid for an arbitrary subsequence of $\bar\pi^*(z_{k})$, we obtain that $\bar\pi^*(z_{k})$ itself must converge to $\bar\pi^*(z_0)$ as $k\to\infty$.

Now, fix a $t\in[0,T]$. Then, the first order condition can be inverted in the following sense:
$$
s^*(t,\bar\pi):=\frac1\lambda \left(b(t)-\sigma(t)^2\bar\pi-\int_\bR 
\frac{\bar\pi \psi(t)^2 y^2}{1+\bar\pi \psi(t) y}\, \nu(dy)\right),
$$ 
so to define the inverse function of $\bar\pi^*(t,\cdot)$ from $(\bar\pi^*(t,\cdot))^{-1}\bigl((\bar\pi_1,\bar\pi_2)\bigr)$ to $(\bar\pi_1,\bar\pi_2)$. By the Inverse Function Theorem (cf. \cite[Appendix C.5]{MR1625845}), since
\begin{equation*}
\frac{\partial s^*(t,\bar\pi)}{\partial \bar\pi} =\frac1\lambda \left(-\sigma(t)^2-\int_\bR 
\frac{\psi(t)^2 y^2}{(1+\bar\pi \psi(t) y)^2}\, \nu(dy)\right)<0,
\end{equation*}
we get that, for a fixed $t\in[0,T]$, $\bar\pi^*(t,\cdot)$ is strictly decreasing and continuously differentiable for any $s\in\bigl(\bar\pi^*(t,\cdot)\bigr)^{-1}\bigl((\bar\pi_1,\bar\pi_2)\bigr)$. Observe, in particular, that $\bigl(\bar\pi^*(t,\cdot)\bigr)^{-1}\bigl((\bar\pi_1,\bar\pi_2)\bigr)$ must be an interval, which we denote by $(s_1(t),s_2(t))$. Also, since $s^*(t,\cdot)$ is smooth by Lebesgue's dominated convergence theorem, with $n$-th derivative  \begin{equation}\label{93}
\frac{\partial^n s^*(t,\bar\pi)}{\partial \bar\pi^n}=(-1)^n\frac{n!}\lambda \int_\bR \frac{ \psi(t)^{n+1} y^{n+1}}{(1+\bar\pi \psi(t) y)^{n+1}}\, \nu(dy),
\end{equation}
then $\bar\pi^*(t,\cdot):(s_1(t),s_2(t))\to(\bar\pi_1,\bar\pi_2)$ is smooth. The boundedness of the derivatives of $\bar\pi^*(t,\cdot)$ follows from \eqref{93} and the fact that, by definition, $1+\bar\pi \psi(t) y\geq\delta>0$ for each $t\in[0,T]$ and $\bar\pi\in\Pi$.

Since $f'(\bar\pi;t,s)\to\mp\infty$ as $s\to\pm\infty$, uniformly with respect to $\bar\pi\in\Pi$ and $t\in[0,T]$, there exist $s_1,s_2$ independent from $\bar\pi$ and $t$ such that, for any $t\in[0,T]$ and $\bar\pi\in[\bar\pi_1,\bar\pi_2]$, $-\infty< s_1\leq \frac{b(t)}\lambda \leq s_2<\infty$ and
$$
\begin{cases} f'(\bar\pi;t,s)>0, & \mbox{if }s\leq  s_1, \\ f'(\bar\pi;t,s)<0, & \mbox{if }s\geq s_2.\end{cases}
$$
By monotonicity of $f$, it follows that
$$
\bar\pi^*(t,s)\equiv\begin{cases} \bar\pi_2, & \mbox{if }s\leq s_1, \\ 
\bar\pi_1, & \mbox{if }s\geq s_2,\end{cases}
$$
which proves the fourth statement. 

The Lipschitz continuity of $\bar\pi^*(t,\cdot)$ follows from the fact that its derivative in $[s_1(t),s_2(t)]$ is bounded uniformly in $t$ and that $\bar\pi^*(t,\cdot)$ is constant outside $[s_1(t),s_2(t)]$ (and the constant is independent of $t$). 

\subsection*{Proof of Theorem \ref{36}}

The admissibility of $\pi^*(u,S(u-),X(u-))=\bar\pi^*(u,S(u-)) X(u-)$ is immediate consequence of Proposition \ref{19}. In order to apply the Verification Theorem (Theorem \ref{verification}), which allows us to conclude, we need to prove that $H(u,s,x)=\log(x)+g(u,s)$ satisfies the Dynkin formula \eqref{4}. By It\^o's lemma, for each admissible strategy $\pi$, we get 
\begin{align*}
dH(u,S(u),X(u))=&\, H_u(u,S(u),X(u))+A^{\pi} H(u,S(u),X(u))\\
&+\sigma(u)\left(\frac{\pi(u)}{X(u)}+ g_s(u,S(u))\right)dW(u)\\
&+\int_{\bR} \Bigl[\log(X(u-)+\pi(u) \psi(u) y)-\log(X(u-))\Bigr]\,\overline N(dy,du)\\
&+\int_{\bR} \Bigl[g(u,S(u-)+\psi(u) y)-g(u,S(u-))\Bigr]\,\overline N(dy,du).
\end{align*}
Since $\pi(u)=\bar\pi(u,S(u-)) X(u-)$, we can rewrite it as
\begin{align*}
dH(u,S(u),X(u))=&\,H_u(u,S(u),X(u))+A^{\pi} H(u,S(u),X(u))\\
&+\sigma(u)\left(\bar\pi(u,S(u)) + g_s(u,S(u))\right)dW(u)\\
&+\int_{\bR} \log(1+\bar\pi(u,S(u-)) \psi(u) y)\,\overline N(dy,du)\\
&+\int_{\bR} \Bigl[g(u,S(u-)+\psi(u) y)-g(u,S(u-))\Bigr]\,\overline N(dy,du).
\end{align*}
Then the validity of Dynkin's formula boils down to the martingale property of the process 
\begin{align*}
dZ(u):=&\,\sigma(u)\left(\bar\pi(u,S(u)) + g_s(u,S(u))\right)dW(u)\\
&+\int_{\bR} \log(1+\bar\pi(u,S(u-)) \psi(u) y)\,\overline N(dy,du)\\
&+\int_{\bR} \Bigl[g(u,S(u-)+\psi(u) y)-g(u,S(u-))\Bigr]\,\overline N(dy,du).
\end{align*}
Sufficient conditions are
\begin{align*}
&\bE\left[\int_{t}^T \sigma(u)^2\, \bar\pi(u,S^{t,s}(u))^2\, du\right]<\infty,\\
& \bE\left[\int_{t}^T \int_\bR \left[\log(1+\bar\pi(u,S(u)) \psi(u) y)\right]^2\,\nu(dy) \, du\right]<\infty,\\
&\bE\left[\int_{t}^T \int_\bR \left[g(u,S^{t,s}(u)+\psi(u) y)-g(u,S^{t,s}(u))\right]^2\,\nu(dy) \, du\right]<\infty,\\
&\bE\left[\int_{t}^T \sigma(u)^2\, g_s(u,S^{t,s}(u))^2\, du\right]<\infty.
\end{align*}
The first two follow from the definition of $\Pi$ (the range of $\bar\pi$) and the boundedness of $\sigma$, while the last conditions depend on $g$ and are assumed valid in the statement. By standard integrability reasonings (see e.g. \cite[Section 3.1]{pasin} for an argument), we get that $Z$ is a martingale, and then the result.

\subsection*{Proof of Proposition \ref{34}}
We follow the same lines of reasoning as in Section 4 of \cite{MR3176490}. 
First, let us write from \eqref{13} for $\pi\in\mathring{\Pi}$ and fixed $t\in[0,T]$ and $s\in\bR$:
$$
h(\pi):=h(\pi;t,s)=b(t)-\lambda s-\pi\sigma(t)^2-\psi(t)^2 \int_\bR \pi y^2\, \nu(dy),
$$
observing that
\begin{align*}
h(\bar\pi^*_1)&=0,\\
h(\bar\pi^*)&=\psi(t)^2\int_\bR 
\frac{\bar\pi^* y^2}{1+\bar\pi^* \psi(t) y}\, \nu(dy)-\psi(t)^2\int_\bR 
\bar\pi^* y^2\, \nu(dy)=-\int_\bR 
\frac{\psi(t)^3 (\bar\pi^*)^2 y^3}{1+\bar\pi^* \psi(t) y}\, \nu(dy).
\end{align*}
Hence,
$$
\bar\pi^*_1=h^{-1}(0),\qquad \bar\pi^*=h^{-1}\left(-\int_\bR 
\frac{\psi(t)^3 (\bar\pi^*)^2 y^3}{1+\bar\pi^* \psi(t) y}\, \nu(dy)\right).
$$
Now,
$$
|\bar\pi^*-\bar\pi^*_1|=\left|h^{-1}\left(-\int_\bR 
\frac{\psi(t)^3 (\bar\pi^*)^2 y^3}{1+\bar\pi^* \psi(t) y}\, \nu(dy)\right)-h^{-1}(0)\right|.
$$
By applying the Mean Value Theorem, since $\bar\pi^*$ takes values in $\Pi=\Pi_{\nu,\psi}$, a compact set whose distance from the boundary of $\widehat\Pi$ is a certain $\delta>0$ and $|(h^{-1}(z))'|=\frac 1{\sigma(t)^2+\sigma_L(t)^2}$, we finally get the above estimate
\begin{align*}
|\bar\pi^*-\bar\pi^*_1|&\leq \frac 1{\sigma(t)^2+\sigma_L(t)^2} \left| \int_\bR 
\frac{\psi(t)^3 (\bar\pi^*)^2 y^3}{1+\bar\pi^* \psi(t) y}\, \nu(dy)\right|\leq \frac{\psi_2^3}{\sigma_1^2+\psi_2^2 \sigma_\nu^2} \int_\bR \frac{(\bar\pi^*)^2 }{1+\bar\pi^* \psi(t) y}\, |y|^3\, \nu(dy)\\
&\leq C \int_\bR |y|^3\, \nu(dy),
\end{align*}
where $C$ is the constant in the statement.

\subsection*{Proof of Proposition \ref{99}}
Define for $\bar\pi\in\mathring{\Pi}$ and fixed $t\in[0,T]$ and $s\in\bR$:
$$
h(\bar\pi):=h(\bar\pi;t,s)=b(t)-\lambda s-\bar\pi\sigma(t)^2-\frac{\bar\pi\psi(t)^2 \mu_F^2}{1+\bar\pi \psi(t) \mu_F} \eta,
$$
being $\eta=\nu([m,M]\setminus\{0\})$.
Then, 
\begin{align*}
h(\bar\pi^*_2)&=0,\\
h(\bar\pi^*)&=\int_\bR 
\frac{\bar\pi^* \psi(t)^2 y^2}{1+\bar\pi^* \psi(t) y}\, \nu(dy)-\frac{\bar\pi^* \psi(t)^2 \mu_F^2}{1+\bar\pi^* \psi(t) \mu_F} \eta,
\end{align*}
and
$$
\bar\pi^*_2=h^{-1}(0),\qquad \bar\pi^*=h^{-1}\left(\int_\bR 
\frac{\bar\pi^* \psi(t)^2 y^2}{1+\bar\pi^* \psi(t) y}\, \nu(dy)-\frac{\bar\pi^* \psi(t)^2 \mu_F^2}{1+\bar\pi^* \psi(t) \mu_F} \eta\right).
$$
Coming back to \eqref{40}, observe that
$$
\int_\bR 
\frac{\bar\pi^* \psi(t)^2 y^2}{1+\bar\pi^* \psi(t) y}\, \nu(dy)-\frac{\bar\pi^* \psi(t)^2 \mu_F^2}{1+\bar\pi^* \psi(t) \mu_F} \eta=\psi(t)^2 \int_\bR (\phi(y)-\phi(\mu_F) )\nu(dy).
$$
By applying the Mean Value Theorem to $\phi$:
$$
|\phi(y)-\phi(\mu_F)|\leq\sup_{z\in[m,M]\cap\bR}|\phi'(z)||y-\mu_F|
\leq C_1|y-\mu_F|,
$$
with, for instance in {\bf case A}, $C_1=\max\{1,\frac{1}{(\delta \psi_2 M)^2},\frac{1}{(\delta \psi_2 m)^2}\}+\max\{1,\frac{1}{\delta \psi_2 M},\frac{1}{-\delta \psi_2 m}\}$,
since (cf. Proposition \ref{34})
$$
|\phi'(z)|=\frac{|\bar\pi z|(2+\bar\pi \psi(t) z)}{(1+\bar\pi \psi(t) z)^2}=
\frac{|\bar\pi z|}{(1+\bar\pi \psi(t) z)^2}+
\frac{|\bar\pi z|}{(1+\bar\pi \psi(t) z)}.
$$
Therefore, we have found the following estimate:
$$
\left|\int_\bR 
\frac{\bar\pi^* \psi(t)^2 y^2}{1+\bar\pi^* \psi(t) y}\, \nu(dy)-\frac{\bar\pi^* \psi(t)^2 \mu_F^2}{1+\bar\pi^* \psi(t) \mu_F} \eta \right|\leq \psi_2^2 \,C_1 \int_\bR |y-\mu_F|\, \nu(dy).
$$
Now,
$$
|\bar\pi^*-\bar\pi^*_2|=\left|h^{-1}\left(\int_\bR 
\frac{\bar\pi^* \psi(t)^2 y^2}{1+\bar\pi^* \psi(t) y}\, \nu(dy)-\frac{\bar\pi^* \psi(t)^2 \mu_F^2}{1+\bar\pi^* \psi(t) \mu_F} \eta \right)-h^{-1}(0)\right|.
$$
Before applying the Mean Value Theorem again, let us compute 
$$
h'(\bar\pi)=-\sigma(t)^2-\frac{\psi(t)^2 \mu_F^2 \eta}{(1+\bar\pi \psi(t) \mu_F)^2},
$$ 
which is negative for each $\bar\pi\in\Pi$, so that
$$
|h'(\bar\pi)|=\sigma(t)^2+\frac{\mu_F^2 \psi(t)^2 \eta}{(1+\bar\pi \psi(t) \mu_F)^2}.
$$ 
Finally, denoting $z=h(\bar\pi)$, since $(h^{-1}(z))'=\frac{1}{h'(h^{-1}(z))}=\frac{1}{h'(\bar\pi)}$,
\begin{align*}
&|\bar\pi^*-\bar\pi^*_2|\leq \sup_{z\in [m,M]\cap\bR}|(h^{-1}(z))'| \left|\int_\bR 
\frac{\bar\pi^* \psi(t)^2 y^2}{1+\bar\pi^* \psi(t) y}\, \nu(dy)-\frac{\bar\pi^* \psi(t)^2 \mu_F^2}{1+\bar\pi^* \psi(t) \mu_F} \eta \right| \\
&\leq \frac {\psi_2^2\,C_1}{\sigma(t)^2+\frac{\mu_F^2 \psi_2^2 \eta}{\sup_{\bar\pi\in\Pi}(1+\bar\pi \psi(t)\mu_F)^2}} 
\int_\bR |y-\mu_F|\, \nu(dy)
\leq \frac {\psi_2^2\,C_1}{\sigma_1^2+ \mu_F^2\,\psi_2^2\, \eta\, C_2} 
\int_\bR |y-\mu_F|\, \nu(dy)
\\
&\leq \frac {\psi_2^2\,C_1\,\eta\,\sigma_F}{\sigma_1^2+ \mu_F^2\, \eta\, \psi_2^2\,C_2},
\end{align*}
where $\sigma_F$ is the square root of the variance of the jump size and $C_1,C_2$ are the constants in the statement.

\bibliographystyle{amsplain} 
\bibliography{bibbase}

\end{document}